\pdfoutput=1
\documentclass[journal,12pt, onecolumn]{IEEEtran}

\usepackage{setspace}
\ifCLASSINFOpdf
\else
  \usepackage[dvips]{graphicx}
  \DeclareGraphicsExtensions{.eps}
\fi
\usepackage{amsthm}

\newtheorem{pr}{Proposition}

\newtheorem{lemma}{Lemma}

\newtheorem{thm}{Theorem}

\ifCLASSINFOpdf
  \usepackage[pdftex]{graphicx}
  \graphicspath{{../pdf/}{../jpeg/}}
   \DeclareGraphicsExtensions{.pdf,.jpeg,.png}
\else
  \usepackage[dvips]{graphicx}
   \graphicspath{{../eps/}}
   \DeclareGraphicsExtensions{.eps}
\fi
\usepackage{tikz}
\usetikzlibrary{shapes, arrows, patterns}

\usepackage{amsmath,amsthm,amssymb}
\usepackage{graphicx}
\usepackage{amsmath}
\usepackage{dsfont}
\usepackage{enumerate}
\usetikzlibrary{calc}
\usepackage{stfloats}

\usepackage[tight,footnotesize]{subfigure}

\usepackage{caption}

\usepackage{multirow}

\DeclareMathOperator{\ee}{\mathbb{E}}			
\DeclareMathOperator{\prob}{\mathbb{P}}			

\begin{document}

%
\title{LORD: Leader-based framework for Resource Discovery in Mobile Device Clouds
}
%
%

\author{Seyed~Mohammad~Asghari,~Yi-Hsuan~Kao,~Mohammad~Hassan~Lotfi, \\Mohammad~Noormohammadpour,~Bhaskar~Krishnamachari, \\
~Babak~Hossein~Khalaj, and Marcos~Katz
\thanks{
S. M. Asghari, Yi-Hsuan Kao, M. Noormohammadpour, Bhaskar~Krishnamachari are with the Department of Electrical Engineering, University of Southern California, Los Angeles, CA, USA
(e-mail:\{asgharip, yihsuank, noormoha, bkrishna\}@usc. edu)}
\thanks{
M. H. Lotfi is with the Department of Electrical and System Engineering, University of Pennsylvania, Pennsylvania, USA (e-mail: lotfm@seas.upenn.edu).}
\thanks{
B. H. Khalaj is with the Department of Electrical Engineering, Sharif University of Technology, Tehran, Iran (e-mail: khalaj@sharif.edu).}
\thanks{
M. Katz is with the Department of Electrical Engineering, University of Oulu, Oulu, Finland (e-mail: mkatz@ee.oulu.fi).
}}
\maketitle

\begin{abstract}
We provide a novel solution for Resource Discovery (RD) in mobile device clouds consisting of selfish nodes. Mobile device clouds (MDCs) refer to cooperative arrangement of communication-capable devices formed with resource-sharing goal in mind. Our work is motivated by the observation that with ever-growing applications of MDCs, it is essential to quickly locate resources offered in such clouds, where the resources could be content, computing resources, or communication resources. The current approaches for RD can be categorized into two models: decentralized model, where RD is handled by each node individually; and centralized model, where RD is assisted by centralized entities like cellular network. However, we propose LORD, a Leader-based framewOrk for RD in MDCs which is not only self-organized and not prone to having a single point of failure like the centralized model, but also is able to balance the energy consumption among MDC participants better than the decentralized model. 
Moreover, we provide a credit-based incentive to motivate participation of selfish nodes in the leader selection process, and present the first energy-aware leader selection mechanism for credit-based models.
The simulation results demonstrate that LORD balances energy consumption among nodes and prolongs overall network lifetime compared to decentralized model.

\end{abstract}



%

\section{Introduction}
According to Gartner \cite{xx1}, smartphone sales have surpassed one billion units in 2014. 
All such devices are typically equipped with multiple types of wireless networking interfaces especially Wi-Fi. By taking advantage of traditional device to device communications (e.g., Wi-Fi direct \cite{x37}, and Bluetooth), devices are able to communicate directly with each other without assistance of any access points or involvement of network operators. 
Throughout this paper, we use the term \textit{Mobile Device Cloud (MDC)} \cite{x334}, \cite{x335} for the cooperative arrangement of these communication-capable devices, which is formed with the resource sharing goal in mind. The resources can include content (e.g., music, photo, and video files) \cite{xx4}, computing resources \cite{xx7}, and communication resources (e.g., communication link to LTE network) \cite{xx3}. In MDC environments, computational offloading is performed by means of a set of mobile devices instead of remote cloud resources. 

Despite the advantages of MDCs, a number of key problems should be addressed before they can be realized in practice. One of the most important issues is resource discovery (RD), i.e., identifying which resources are offered in such MDCs and also how to obtain them. 
As an example, consider a situation where a client wants to solve engineering problems or process photos and video files, yet it is unable to perform this task individually either due to lack of enough battery life or because of the fact that the computing requirements of such a task outstrip what his device is able to accomplish. However, there are often a large number of nearby smartphones which remain unused most of the time. Smartphones now have powerful processors (e.g., Dual-core 1.4 GHz Cyclone and PowerVR GX6450 of iPhone 6), powerful operating systems (e.g., Apple iOS and Google Android), remarkable battery life, and plentiful memory that make them appropriate for processing tasks.
This client can offload his task to one (or more) near device(s) and hence, it needs to discover which devices can offer computing resources.

The current approaches for RD can be categorized into two models: decentralized and centralized. 
In the decentralized model, each node is responsible for RD itself; however, in the centralized one, RD is handled by one server. As an example of centralized model, a network-assisted device to device communication has been proposed in \cite{x53} where LTE network can act as the server.
Although centralized models have significant advantages in comparison with decentralized one (e.g., lower communication overhead, lower RD delay, more suited to dynamic environments \cite[and references therein]{x14}), they suffer from the important drawback of having a critical point of failure \cite{x14}. In case of the model proposed in \cite{x53}, such failures can be due to the fact that the devices have no access to the LTE network, or the Base Transceiver Station (BTS) operation is disrupted.  

The centralized model does not suit very well MDC environments that lack a predefined organization and where accessibility of nodes to an external server (e.g. via LTE connection) can be intermittent. On the other hand, considering the advantages of centralized model in comparison with decentralized one, task-based self-organizing algorithms can be used to address the point of failure problem by sequentially selecting the best nodes to act as servers \cite{x34}. In \cite{x15}, authors proposed a management framework for MDCs in which a node with highest level of resources is selected as leader. Such a leader is responsible for introducing a suitable resource provider to each RD requesting node. However, such a framework suffers from the key disadvantage that it assumes each node reports its resource level truthfully, which is not necessarily true in the presence of selfish nodes. Since the leader consumes more energy compared to the other nodes in the MDC, nodes have no inherent incentives for accepting the leadership role. 

Although the presence of selfish nodes in the packet forwarding problem has been addressed well in the literature (e.g., \cite{x20}, \cite{x21}), leader selection and RD in the presence of such nodes is not a well-studied topic
\cite{Mohammad_2013_malicious}. In \cite{x22}, in order to provide selfish nodes with incentives in the form of reputations, the authors have proposed a secure leader selection model for intrusion detection systems based on the Vickrey, Clarke, and Groves (VCG) mechanism. 
In addition to the fact that there is no formal specification and analysis of the type of
incentive provided by reputation-based systems, the use of price-based systems is more advantageous for RD in MDC since they are more secure than reputation-based ones \cite{x23}. However, 
the model in \cite{x22} cannot be applied to price-based systems. That is due to the fact that VCG mechanism is not budget-balanced\footnote[1]{For a mechanism which requires transfer of money among the nodes, Budget-Balanced means that the net transfer among nodes is zero. Hence, there is no need to inject money to the system \cite{x24}.} and it needs money to be injected to the system which is not suitable \cite{x2200}. 
In \cite{Mohammad_2013_leader}, the authors proposed a leader selection process which is based on a series of one-on-one incomplete information, alternating offers bargaining games. However, the proposed leader selection process is not appropriate for MDC environment because of the messaging overhead it imposes on the MDC.
Therefore, a new leader selection model compatible with price-based systems is required.

We list our main contributions as follows:
\begin{enumerate}
\item \textbf{LORD, a Leader-based framewOrk for RD:} 
We propose LORD which is not only self-organized and not prone to having a single point of failure like the centralized model, but also is able to balance the energy consumption among MDC participants better than the decentralized model. In LORD, to encourage nodes to assume leadership role, incentives are designed in the form of credit transfer from RD requesting nodes to the leader as an RD fee. 

\item \textbf{The first leader selection mechanism for credit-based systems:}
We propose a leader selection mechanism which is based on first-price sealed-bid auction. The pivotal difference between our mechanism and preceding ones is that we are seeking two objectives simultaneously: \textit{sequentially selecting nodes with the highest remaining energy as the leader based on proposed credit transfer model} and \textit{determining an appropriate RD fee that is lucrative for both the leader and RD requesting nodes}.

\item \textbf{Comparative performance evaluation:}
We evaluate the performance of LORD through simulation. Compared to decentralized model for RD, simulation results demonstrate that LORD balances energy consumption among nodes in the MDC and prolongs overall network lifetime. Furthermore, nodes enjoy higher payoffs in LORD compared to decentralized model.

\end{enumerate}

\section{The System Architecture}
We consider a mobile device cloud (MDC) in which all participating nodes have demand for RD. If a node is selected as the leader, the cost imposed on this node is expressed in its Resource Discovery Provisioning Cost (RDPC), which will be defined in Section \ref{section:parameters}. In LORD, we select the node with the lowest RDPC as the leader. 
To encourage nodes to assume leadership role, we consider incentives in the form of credit transfer from RD requesting nodes to the leader. 

In order to select the node with the lowest RDPC as the leader, we propose a leader selection mechanism which is based on multi-player first-price sealed-bid auction. In the leader selection process, each node offers an RD fee in order to compensate for the RDPC imposed on it. The amount of the RD fee offered by each node reflects the private value of its RDPC. Then, the node with the lowest offered RD fee is selected as the leader, which minimizes the total RD cost of participants in the MDC. 

As the resource level of nodes decreases over time, their RDPC values will increase. Since the amount of payment offered by each node in the leader selection process is on the basis of the current value of RDPC, at some point, the amount of payment may not be enough to compensate its RDPC. In order to address this issue, each node assumes leadership role for one time-slot that lasts $T_{select}$ units of time.
In order to implement a multi-player auction among participating nodes, LORD consists of two different phases as illustrated in Fig. \ref{figure 5}.

\begin{figure}[b]
\begin{center}
\includegraphics[width=2.8in,keepaspectratio]{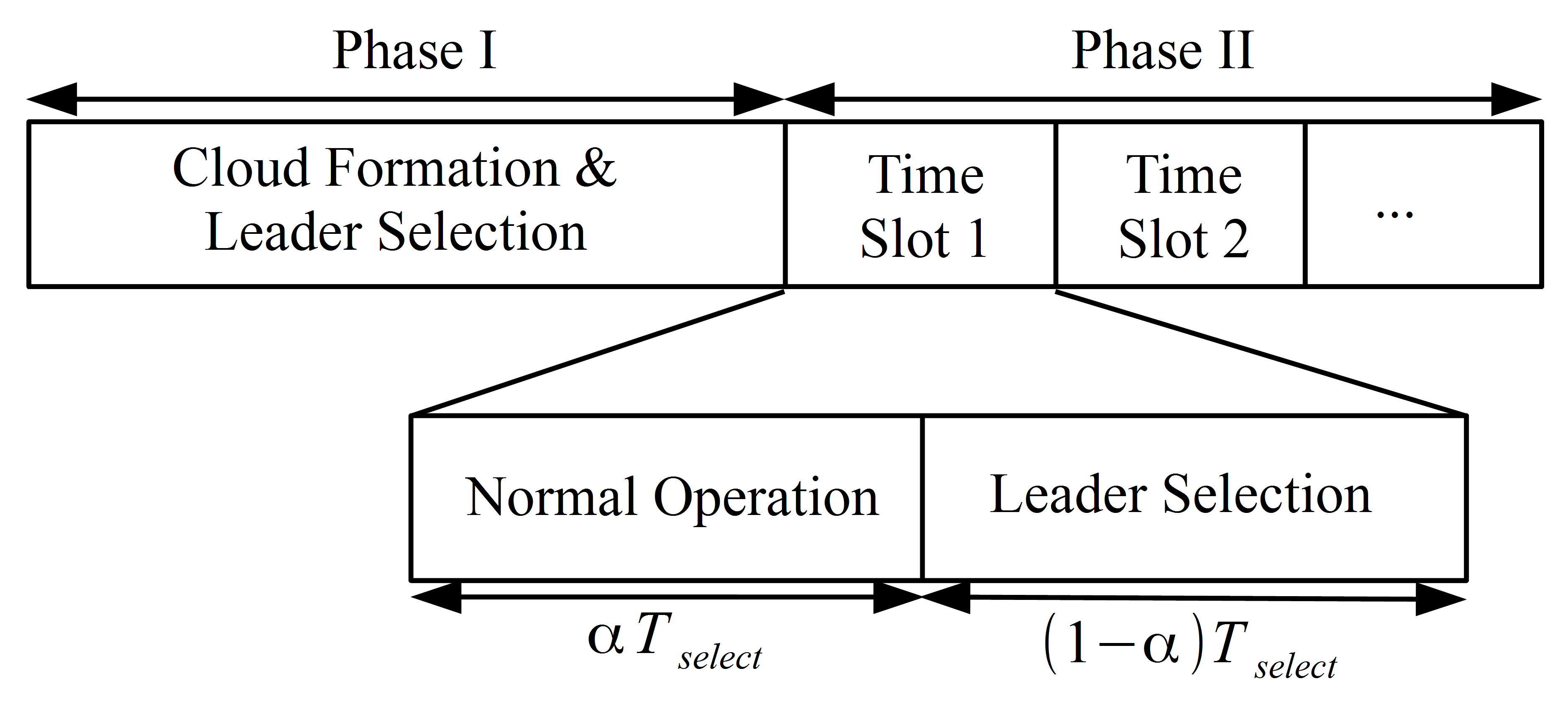}
\caption{LORD timing for MDC}
\label{figure 5}
\end{center}
\end{figure}

The aim of the first phase is to form a MDC including all participants, and to select the node with the lowest RD fee as its leader. Since there is no central entity known to all nodes to perform a multi-player auction, 
we perform the auction as a series of two-node interactions. In each two-node interaction, the node that has the lowest RD fee is selected as the leader, and another node is considered as a client. Then, the interactions are performed one-on-one among the remaining leaders. The first phase ends when just one node is left which is the leader. This node is recognized as the leader of the MDC by others, i.e., clients.

In the second phase, the leader selected in the first phase assumes the leadership role during the first time-slot. Such a leader is responsible for classifying the available resource status data in a database based on the information each node sends regularly to it. After receiving an RD request from a client, the leader uses the information available in the database to find a proper resource provider, and then receives a voucher from the RD requesting node in return. 

The value of each voucher is equal to the amount of the RD fee it has offered in the leader selection process. Afterwards, the RD requesting node and the resource provider reach an agreement over the protocols and other requirements of the resource sharing such as resource price; a process which is beyond the scope of this paper. Finally, after a time equal to $\alpha T_{select}$, where $\alpha$ is slightly smaller than 1, a new leader should be selected for the next time-slot. Since current leader of the MDC is known to clients, such a node performs an auction among all participants of the MDC. The winner of the auction assumes leadership role for the next time-slot.

We also consider the presence of one Trusted Unit (TU) which is only responsible for credit transfer among participants \cite{x23} and also can act as a mediator \cite{x500} to avoid misbehavior of nodes. It should be noted that since our objective is to design a framework which is independent of permanent presence of any unit outside of the MDC, the TU cannot act as the entity performing the auction among participating nodes. In general, any trusted entity outside the MDC which is intermittently accessible by all members, can be considered as TU. 

In LORD, a specific amount of credit is assigned to each node entering the MDC for the first time. In addition, each node can earn credits by providing resource for other nodes. In return, they can spend such credits to get resource from other nodes.
A typical architecture of LORD is shown in Fig. \ref{figure 2}. One node plays the role of leader while others act as clients. When client 3 sends a request to the leader, the leader finds a proper resource provider among other nodes, which is client 5 in this case. Then, the leader introduces client 5 to client 3, and receives a voucher as RD fee in return. Finally, after reaching an agreement over the resource sharing requirements, client 5 provides the resource for client 3 and receives some credit in return. 
\begin{figure}[t]
\begin{center}
\includegraphics[width=2.8in,keepaspectratio]{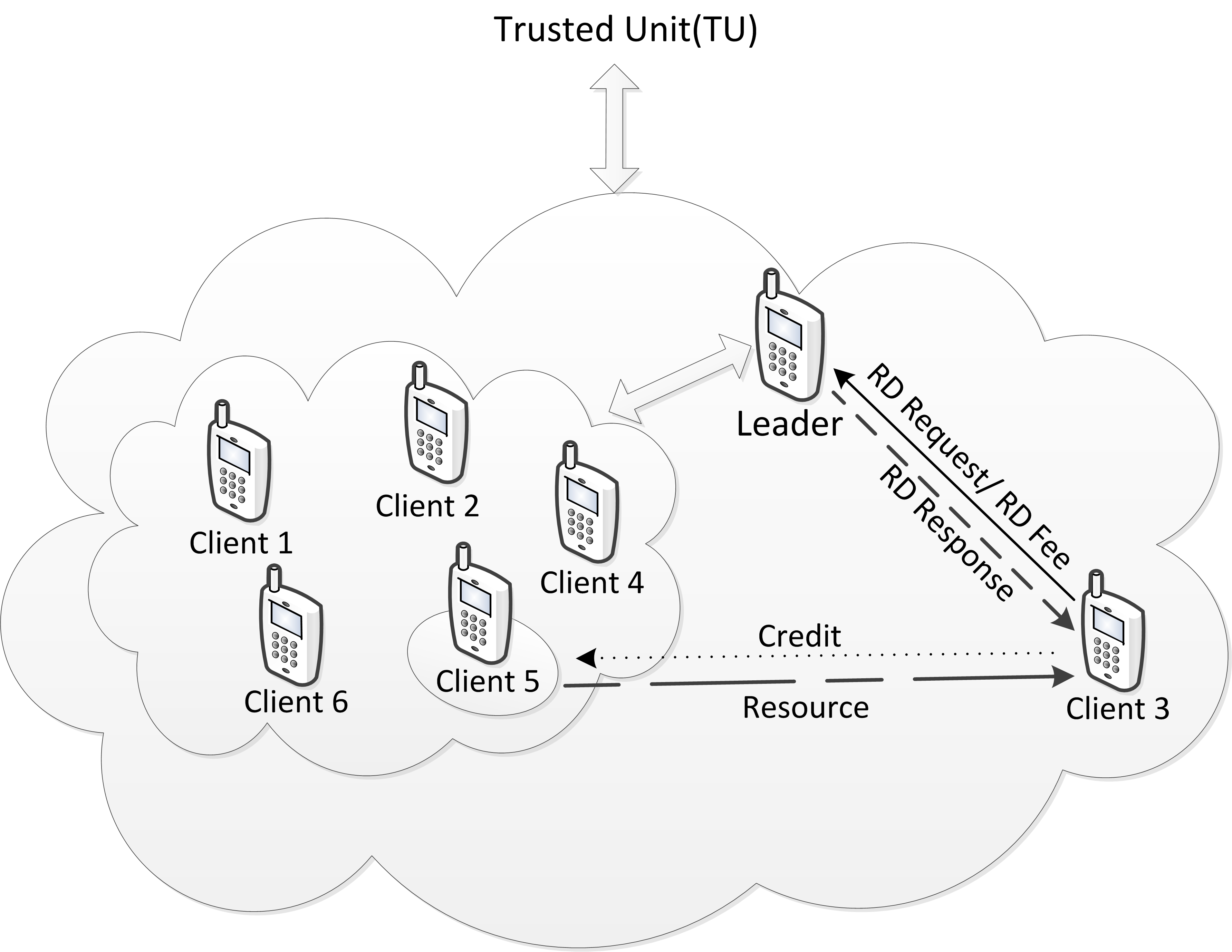}
\caption{System architecture for LORD}
\label{figure 2}
\end{center}
\end{figure}

\section{LORD: Leader-Based Framework for Resource Discovery} \label{section:model}

In this section, we present LORD for mobile device clouds (MDCs). First, we state our assumptions (Section~\ref{section:assumptions}). Then parameters of LORD are introduced in Section~\ref{section:parameters}. Section~\ref{section:Leader Selection Game} formulates the leader selection problem. Then, the two phases of LORD are introduced (Sections~\ref{section:Phase1} and \ref{section:phase2}). Finally, we consider the event of entrance and departure of nodes to/from the MDC.   

\subsection{Assumptions}\label{section:assumptions}
A MDC consisting of \textit{n} nodes is considered where each node has a unique identity. This MDC can be formed in various ways including manual assignment (similar to adding all iOS devices to one's iTunes, or android devices with Google), intelligent device
contact profiling, or via social profiles present on these devices (e.g. devices of members in the same household might be assigned to a single MDC). Such discovery and cloud clustering process can be addressed in various ways in the research
community \cite{x27} and therefore, we assume that such mechanisms exist \cite{x335}. Furthermore, we assume that the credit of all nodes are stored in a secure manner by the TU. In order to get credit from RD requesting nodes, the leader needs to report the vouchers it has received to the TU. It can save its vouchers inside a local storage such as a CompactFlash card, or a secure element such as a Subscriber Identification Module (SIM) card. Vouchers should be reported to the TU when the nodes have access to it.

In addition, since RDPCs of nodes should be calculated prior to beginning of each time-slot, the value of $\eta$ is fixed. We assume that the demand of participating nodes for RD in the MDC is high, but the leader responds to only $\eta$ number of RD queries.

Furthermore, we consider the following assumptions: 
\begin{itemize}
\item Every MDC participant is aware of the number of nodes participating in the leader selection process\footnote[1]{
Note that the duration of each leader selection process is negligible. Furthermore, each node maintains a table for routing purposes. Therefore, this assumption is reasonable because the number of nodes participating in the leader selection could be obtained easily from routing tables.}.
\item Time-slot synchronization is available among all nodes. 
\item A reliable pairwise connection among nodes is established at the transport layer. Therefore, all routing and error-correcting operations are handled by lower network layers and are hidden to the procedures we explain. As a result, there is an error-free route available between each two nodes. 

\end{itemize}
It should be noted that the first two assumptions have also been implemented in \cite{x22}.
\subsection{Parameters and Concepts}\label{section:parameters}

In LORD, one node is selected as the leader among $n$ participants. The selected leader has the responsibility of RD for a given time-slot. All clients should send the information about the status of their resources to the leader at the beginning of each time-slot. The information each node sends to the leader should represent whether this node can provide $RT_{k}$ (Resource Type $k$) and the amount of money (resource price) it expects to receive if it can provide such a resource. Based on such information, the leader prepares a database in which for each $RT_{k}$, it has been specified which nodes are able to provide this resource and how much their corresponding resource prices are. The leader uses this database to find the best resource provider, i.e., the one with the lowest resource price, for each RD request. Furthermore, the leader is responsible for selecting a new leader for the next time-slot.

If node $i$ is the leader and receives $p^*$ for each RD, its payoff will be as follows:
\begin{align}
\label{equation:xx1}
U_i &= (n-1)\eta(p^* - c_i^{S}- c_i^{I}) - c_i ^{DB}-m_i c_i^{S} - c_i ^{LS}
\end{align}
where $\eta$, as mentioned in Section~\ref{section:assumptions}, is the number of RD requests of each client to which the leader responds, $c_i^{S}$ is the processing cost imposed on node $i$ for searching the database in order to find a proper resource provider, $c_i^{I}$ is the messaging cost imposed on node $i$ for introducing such a resource provider to the RD requesting node, $c_i^{DB}$ is the processing cost imposed on node $i$ for preparing a database, and $c_i^{LS}$ is the messaging cost imposed on node $i$ for selecting a new leader among all the participants.  Finally, $m_i$ is the number of times this node searches the database for finding a resource provider for itself.

%

Simplifying (\ref{equation:xx1}), the leader's payoff can be rewritten as,
\begin{equation}
\label{equation:xx2}
U_i =  M(p^* - c_i)
\end{equation}
where $M$ is the total number of RDs the leader should perform for clients during one time-slot, i.e., $ M=(n-1)\eta$, and $c_i$ is the Resource Discovery Provisioning Cost (RDPC) per RD imposed on the leader given by:
\begin{align}
\label{equation:xx3}
c_i = c_i ^{S} +c_i ^{M} + \dfrac{c_i ^{DB} + c_i ^{LS} + m_i c_i ^{S}}{M}
\end{align} 

It should be noted that $c_i^{DB}$ and $c_i^{S}$ are the processing costs of preparing and searching a database which hinges on the processing power of node $i$, and the number of entries in the database (which depends on the number of nodes ($n$)). Therefore, although the database has not been constructed when nodes want to set their RDPC, since the number of nodes is specific (using the information of routing tables), $c_i^{S}$ and $c_i^{DB}$ can be readily calculated.
For simplicity of analysis, we assume these two costs are equal to $c_i^{P}$ which represents the processing cost imposed on node $i$. Furthermore, the term $c_i^{I}$ is a messaging cost requiring only 3 messages (as will be explained in Section \ref{trans_protocol}), and the term $c_i^{LS}$ is also a messaging cost requiring $2(n-1)$ messages (as will be explained in Section \ref{section:phase2}). If we define $c_i^{M}$ as the cost imposed on node $i$ for sending a message, $c_i^{I}$ and $c_i^{LS}$ will be equal to $3c_i^{M}$ and $2(n-1)c_i^{M}$, respectively. Therefore, (\ref{equation:xx3}) can be simplified as follows:
\begin{align}
\label{equation:xx333}
c_i = (1 + \dfrac{m_i +1}{M})c_i ^{P} + (3+\dfrac{2}{\eta})c_i ^{M}
\end{align} 

At each time-slot, each client should send a message to the leader representing which resources this node can provide and how much money (resource price) it expects to receive if it can provide such a resource. Subsequently, if a client intends to receive one RD from the leader, it should send one RD request imposing $c_i ^{M}$ on this node and also pay $p^*$ to the leader. Assume that the valuation of each RD for a client is fixed and high. Therefore, such a valuation does not affect the strategy of the client. If client $i$ intends to receive $\eta$ number of RDs from the leader, the payoff of this node is defined as:
\begin{align}
\label{equation:xx4}
U_i = -\eta p^* -\eta c_i ^{M} - c_i ^{M} 
\end{align}
where the first term of the right-hand side of (\ref{equation:xx4}) is the monetary cost imposed on this client for $\eta$ number of RDs, the second term is the messaging cost imposed on this client for sending $\eta$ number of RD requests to the leader, and the third term is the cost of sending one message including the information about the status of this client's resources to the leader.

On the other hand, in decentralized scheme for RD, there is no leader. Therefore, each node performs RD itself. Similar to our model which introduces the best resource provider to each RD requesting node, if node $i$ desires to find the best resource provider for $RT_{k}$ itself, it should send messages to all other $n-1$ nodes one-by-one and ask them whether they can provide $RT_{k}$ and inquire about the amount of money (resource price) they expect to receive if they can provide such a resource. After sending $n-1$ messages which impose $(n-1)c_i ^{M}$ cost on node $i$, this node is able to find the best resource provider for $RT_{k}$, i.e., the one with the lowest resource price. Therefore, if node $i$ intends to perform $\eta$ number of RDs itself, its payoff is given by: 
\begin{align}
\label{equation:xx5}
U_i = -\eta (n-1)c_i ^{M} =  -\eta \overline{c}_i
\end{align}
where $\overline{c}_i$ is the Resource Discovery Cost (RDC) per RD imposed on node $i$ in a decentralized scheme.

In this paper, we assume RDPC of each node is smaller in comparison with its RDC, i.e., $c_i < \overline{c}_i \hspace{3mm} \forall i$. Since RDPC and RDC depend on $c_i ^{P}$ and $c_i ^{M}$, specifying the precise value of RDPC and RDC needs experimental data which is beyond the scope of this paper. It can be shown that if $c_i ^{P} \leq c_i ^{M}$ (which is plausible since $c_i ^{P}$ is a negligible cost due to high processing power available in today's smartphones \cite{x29}), our assumption that $c_i < \overline{c}_i$ is valid (See Appendix \ref{validating}).

By comparing (\ref{equation:xx4}) and (\ref{equation:xx5}), it can be concluded that the client $i$ participates in LORD if, 
\begin{align}
\label{equation:xx668}
&-\eta p^* -(\eta +1) c_i ^{M} > -\eta (n-1)c_i ^{M} \Longrightarrow
\nonumber \\
&  p^* < (n-2)c_i ^{M} - \dfrac{1}{\eta} c_i ^{M} < (n-1)c_i ^{M} \Longrightarrow p^* < \overline{c}_i
\end{align} 

Assume that every node $i$ bids an RD fee ($b_i$) in the interval $[c_i,\overline{c}_i]$ and proportional to its RDPC (i.e., $b_i <\overline{c}_i$ and the higher the RDPC, the higher the RD fee). If node $i$ with the lowest RD fee is selected as the leader, this node not only has the lowest RDPC, but also has the lowest RDC (i.e., $\overline{c}_i<\overline{c}_j \hspace{3mm}  j\neq i$). This is due to the fact that both RDPC and RDC depend on the resource status of a node and have a strictly decreasing relation with the resources status, i.e., if nodes are arranged based on their RDPCs and RDCs, all nodes have the same ranking position in both arrangements. 
If node $k$ has the lowest RD fee, i.e., $p^* = \min_{\forall j} b_j = b_k$, this node has the lowest RDPC and RDC, i.e., $\overline{c}_k <\overline{c}_i \hspace{3mm} \forall i\neq k$. Therefore,
\begin{align}
\label{equation:xx601}
b_k <\overline{c}_k  \hspace{3mm} \&  \hspace{2mm} \overline{c}_k<\overline{c}_i \hspace{3mm} \forall i\neq k \hspace{3mm} \Longrightarrow \hspace{3mm} b_k <\overline{c}_i \hspace{3mm} \forall i
\end{align} 
which satisfies the condition of (\ref{equation:xx668}).


Let $E_i ^{req}$ denote the minimum level of energy required for node \textit{i} in order to assume the leadership role for one time-slot, then by considering (\ref{equation:xx2}) and (\ref{equation:xx333}),
\begin{align*}
E_i^{req} = \big((n-1)\eta + m_i +1\big) c_i ^{P} +  \big( 3(n-1)\eta+2(n-1) \big) c_i ^{M}  
\end{align*} 
Subsequently, each node can set its RDPC as follows: 
\begin{align}
\label{equation:xx8}
    c_i = \begin{cases}
               \infty   & E_i \le E_i^{req}\\
             (1 + \dfrac{m_i +1}{M})c_i ^{P} + (3+\dfrac{2}{\eta})c_i ^{M} & E_i > E_i^{req}
           \end{cases}
\end{align}
where $E_i$ is the current level of energy for the node \textit{i}. Based on (\ref{equation:xx8}), a node has a very high RDPC if its remaining energy is less than the energy required to assume leadership role for one time-slot.

\subsection{Auction-Based Leader Selection}
\label{section:Leader Selection Game}
Since the RDPC values of nodes are private to every node, we model the leader selection problem by an incomplete information game known as Bayesian Game. The leader selection problem is considered as a multi-player first-price sealed-bid auction game, where players are the nodes that participate in the MDC denoted by $N=\{1, 2,..., n\}$. The payoff of winner and losers are the payoff of leader node and client node, respectively. Since our objective is to compare the payoff of nodes in LORD with the corresponding payoffs in decentralized scheme for RD, in order to simplify our analysis, we eliminate $-(\eta+1) c_i ^{M}$ in (\ref{equation:xx4}) and rewrite (\ref{equation:xx2}), (\ref{equation:xx4}), and (\ref{equation:xx5}). By doing so, if node \textit{i} is the leader of LORD, its modified payoff will be as follows: 
\begin{align}
\label{equation:xx222}
\widehat{U}_i &=  M\Big(p^* -(1 + \dfrac{m_i +1}{M})c_i ^{P} - (3+\dfrac{2}{\eta})c_i ^{M}\Big) + (\eta +1) c_i ^M 
\nonumber
\\
&=M(p^* - \widehat{c}_i) 
\end{align}
where $\widehat{c}_i$ which is called modified RDPC is given by,
\begin{equation}
\label{equation:xx2232}
\widehat{c}_i = (1 + \dfrac{m_i +1}{M})c_i ^{P} + (3+\dfrac{2}{\eta} - \dfrac{\eta +1}{M})c_i ^{M}
\end{equation}
If node \textit{i} is a client in LORD, its modified payoff will be as follows:
\begin{equation}
\label{equation:xx223}
\widehat{U}_i =  -\eta p^* 
\end{equation}
Finally, if node \textit{i} participates in a decentralized scheme for RD, its modified payoff is given by:
\begin{equation}
\label{equation:xx224}
\widehat{U}_i =  -\eta (n-1)c_i ^M + (\eta +1) c_i ^M = -\eta (n-2-\dfrac{1}{\eta})c_i^M = -\eta \widehat{\overline{c}}_i
\end{equation}
where $\widehat{\overline{c}}_i$ which is called modified RDC is given by,
\begin{equation}
\label{equation:xx2233}
\widehat{\overline{c}}_i = (n-2-\dfrac{1}{\eta})c_i^M 
\end{equation}

In our game, the type of a player is its modified RDPC, i.e., $\widehat{c}_i$, which is private. The strategy of each player is the bid it submits as its RD fee. Such a parameter denoted by $b_i$ for node $i$ is a function of $\widehat{c}_i$, i.e., $b_i=\sigma_i (\widehat{c}_i)$ where $\sigma_i$ is an increasing, continuous, and differentiable function. We will later prove that the Nash Equilibrium (NE) solution indeed satisfies such an assumption. Assume that the belief of each player about the modified RDPC of others is independent of the player, and is  uniformly distributed in the interval $[0, K]$. Furthermore, the payoff of winner and losers of the game are (\ref{equation:xx222}) and (\ref{equation:xx223}), respectively. 

If player $i$ bids $b_i$, its expected payoff is,

\begin{small}
\begin{align} 
\label{equation:xx9}
 &\widehat{U}_i \big( \widehat{c}_i , b_i ;\{\sigma _k(\widehat{c}_k)\}_{k \neq i} \big)=
\nonumber \\
&
\ee [M(b_i - \widehat{c}_i)|b_i< \sigma _k(\widehat{c}_k) \hspace{2mm} \forall k \neq i ] \prob \big(b_i< \sigma _k(\widehat{c}_k) \hspace{2mm} \forall k \neq i\big) 
\nonumber
\\ & + \ee [-\eta\sigma _j(\widehat{c}_j)|\sigma _j(\widehat{c}_j)< b_i, \sigma _j(\widehat{c}_j)< \sigma _l(\widehat{c}_l) \hspace{2mm} \forall l \neq i,j ] 
\nonumber
\\ &
\times \prob \big(\sigma _j(\widehat{c}_j)< b_i, \sigma _j(\widehat{c}_j)< \sigma _l(\widehat{c}_l) \hspace{2mm} \forall l \neq i,j\big) 
\end{align}
\end{small}
where $\prob(.)$ and $\ee[.]$ denote the probability of an event and the expectation of a random variable, respectively. In (\ref{equation:xx9}), the term $M(b_i - \widehat{c}_i)$ is the payoff of player $i$ if it wins the auction, and the term $-\eta\sigma _j(\widehat{c}_j)$ is its payoff when it loses the auction and another player (without loss of generality call it $j$) wins the auction. 

We investigate the class of symmetric NEs, i.e., $\sigma_j (\widehat{c}_j)=\sigma (\widehat{c}) \hspace{3mm} \forall j$. Utilizing the notion of symmetric equilibria helps us reduce the multiplicity of NE. The payoff of player $i$ is,

\vspace*{-3mm}
\begin{small}
\begin{align} 
&\widehat{U}_i \big(\widehat{c}_i , b_i ;\{\sigma(\widehat{c}_k)\}_{k \neq i}\big)=
\nonumber \\
& M(b_i - \widehat{c}_i) \int ...\int _ { \begin{array}{c}
\widehat{c}_k>\sigma ^{-1}(b_i) \\
\forall k \neq i \end{array}} \Big( \frac{1}{K}\Big)^{n-1} d\widehat{c}_1 ...d\widehat{c}_k \hspace{2mm} 
\nonumber
\\ &-\int_{
\widehat{c}_j<\sigma ^{-1}(b_i)}
\frac{\eta\sigma (\widehat{c}_j) }{K}\int ...\int _ {\begin{array}{c}
\widehat{c}_l>\sigma^{-1}\big(\sigma(\widehat{c}_j)\big)\\
\forall l \neq i,j  \end{array}} \Big(\frac{1}{K}\Big) ^{n-2} d\widehat{c}_1 ...d\widehat{c}_l d\widehat{c}_j 
\nonumber
\end{align}
\end{small}

Note that the distribution of modified RDPCs are independent from each other, therefore the bids are independent. Thus,
\begin{align} 
\widehat{U}_i \big(\widehat{c}_i , b_i &;\{\sigma(\widehat{c}_k)\}_{k \neq i}\big)
=
M(b_i - \widehat{c}_i) \Big(\int_ {\widehat{c}_k>\sigma ^{-1}(b_i)}\frac{1}{K}d\widehat{c}_k \Big) ^{n-1}
\nonumber \\
& -\int_{\widehat{c}_j<\sigma ^{-1}(b_i)}\dfrac{\eta\sigma (\widehat{c}_j) }{K} \Big( \int _ {\widehat{c}_l>\widehat{c}_j} \frac{1}{K} d\widehat{c}_l \Big) ^{n-2}d\widehat{c}_j
\nonumber  \\ 
 &=  M(b_i - \widehat{c}_i) \Big( \frac{K - \sigma ^{-1}(b_i)}{K}\Big) ^{n-1} 
\nonumber \\
& -\int_{\widehat{c}_j<\sigma ^{-1}(b_i)}\frac{\eta\sigma (\widehat{c}_j) }{K}\Big( \frac{K- \widehat{c}_j}{K}\Big) ^{n-2}d\widehat{c}_j 
\end{align}

Player $i$ submits $b_i$ that maximizes its expected payoff. The first order necessary condition (with respect to $b_i$) is,
\begin{align} 
\frac{\partial{\widehat{U}_i}}{\partial b_i} &= M\big(K - \sigma ^{-1}(b_i)\big) ^{n-1} 
\nonumber \\ 
&
+ M(b_i - \widehat{c}_i)(n-1)\big(K - \sigma ^{-1}(b_i)\big)^{n-2}\Big(-\sigma ^{'}\big(\sigma^{-1}(b_i) \big)\Big)^{-1}
\nonumber \\ 
&- \Big(\sigma ^{'}\big(\sigma^{-1}(b_i) \big)\Big)^{-1} \eta\sigma \big(\sigma^{-1} (b_i)\big)\big(K- \sigma^{-1} (b_i)\big)^{n-2}=0
\nonumber
\end{align}

After simplification, Bayesian Nash equilibrium of the auction must satisfy the following differential equation:
\begin{align}
\label{equation:x13}
M(K-\widehat{c})\sigma^{'}(\widehat{c}) - M \Big(\sigma(\widehat{c}) -\widehat{c}\Big)(n-1) -\eta\sigma(\widehat{c}) =0
\end{align}

The following Lemma provides us with the solution to (\ref{equation:x13}).
\begin{lemma}\label{lemma:diff}
The following strategy solves the differential equation given by (\ref{equation:x13}):
\vspace{-2mm}
\begin{align}
\label{equation:x14}
b = \sigma (\widehat{c})= \frac{(n-1)^2}{n(n-1)+1}\big(\widehat{c} + \frac{K(n-1)}{(n-1)^2+1}\big)
\end{align} 
 
\end{lemma}
See Appendix \ref{p_l1} for a proof.

First note that, the solution $\sigma_i$ is indeed an increasing, continuous, and differentiable function. Furthermore, (\ref{equation:x14}) implies that each node offers an RD fee which is quite close to its modified RDPC for large enough $n$. Therefore, if the value of $\eta$ in (\ref{equation:xx1}) is selected properly, we can conclude that for every $i$, $b_i < \overline{c}_i$ which ascertains (\ref{equation:xx601}). In the following Proposition, we derive a condition on the value of $\eta$ by which $b_i < \overline{c}_i$ will be satisfied for every~$i$.

\begin{pr}\label{peoposition:etha}
In order to satisfy the condition $b_i < \overline{c}_i \hspace{2mm} \forall i$, the minimum number of RD queries done by the leader for each client in a MDC of size $n$ should be,
$$
\eta_{min}(n)=\max_{RS_i\in [RS_{min}, RS_{max}]}\eta_{min}(n,RS_i,\lambda) 
$$
where $\lambda$ indicates the upper bound of $m_i$ for all nodes and $RS_i$ denotes the resource status of node $i$. Furthermore, $RS_{min}$ and $RS_{max}$ are the lowest and highest values of resource status of all nodes, respectively.
\end{pr}
See Appendix \ref{p_p1} for a proof.
In the numerical result presented in Appendix \ref{p_p1}, we justify that when $n > 10$ and $c_i^M \geq 10 c_i^P$, $\eta_{min} < 1$.
Hence, there is no constraint on $\eta_{min}$ and $b_i < \bar{c}_i$ holds. 

Based on the next Theorem, the strategy identified in (\ref{equation:x14}) is indeed the unique NE for the game. 

\begin{thm}\label{peoposition:sufficiency}
In our proposed multi-player auction, the strategy identified in (\ref{equation:x14}) is the unique symmetric Bayesian Nash equilibrium.
\end{thm}

In the proof, we will argue that if all players except a particular player, say $i$, opt the strategy identified in (\ref{equation:x14}), then the unique strategy which maximizes the payoff of player $i$ is (\ref{equation:x14}). For details of proof, see Appendix \ref{p_t1}.

\begin{thm}
Using the aforementioned leader selection mechanism and considering the assumption of $c_i^P \leq c_i^M$, the payoffs of MDC participants are higher in LORD compared to decentralized model.
\end{thm}

\begin{proof}
The proof is simply resulted from the fact that our proposed leader selection mechanism selects the node with the lowest RD fee. Therefore, according to (\ref{equation:xx601}), the condition of (\ref{equation:xx668}) is satisfied. 
\end{proof}

\begin{pr}\label{energy_saving}
Total consumed energy of nodes is lower in LORD compared to decentralized model if the following condition holds:
\vspace{-2mm}
\begin{align*}
\min_i c_i^M < \dfrac{1- \frac{\eta +1}{\eta(n-1)}}{5+ \frac{2}{\eta} } \sum_{j=1}^n c_j^M
\end{align*}
\end{pr}
See Appendix \ref{p_p2} for a proof. In the numerical result presented in Appendix \ref{p_p2}, we justify that when $n \geq 10$ and $c_i^P \leq c_i^M$, the condition of Proposition \ref{energy_saving} holds. 

In the following subsections, we describe different phases of LORD.

\subsection{Phase I}\label{section:Phase1}
The goal of the first phase is to form a MDC including all participants, and to select the node with the lowest RD fee as its leader. In order to implement a multi-player auction among all participants in a way that does not impose significant communication overhead on the network, we implement the auction as a series of two-node interactions.  
First, nodes calculate their bids for RD fee based on (\ref{equation:x14}). Then, they interact with each other in a one-on-one basis. The node that has lower RD fee is considered as the winner of the interaction, and the other node is considered as the loser. In order to implement a simultaneous interaction between two nodes in which no one is able to change its RD fee after observing RD fee of another one, we define the interaction between nodes $i$ and $j$ as follows:
\begin{enumerate}
\item One of the nodes (e.g., node $i$) sends a MDC formation request including $MD(b_i)$ to node $j$, where $b_i$ is the bid of node $i$ and $MD(.)$ stands for a Message Digest function such as MD5 \cite{x30} or SHA-1 \cite{x31}.
\item When node $j$ receives such a request, if it is a client, it forwards the request to its leader. Otherwise, it submits its bid, i.e., $b_j$ to node $i$. Since node $j$ cannot extract $b_i$ from $MD(b_i)$, its offer is independent of $b_i$.
\item Upon receiving such a response, node $i$ submits its bid, i.e., $b_i$ to node \textit{j}. Since node $i$ has previously submitted $MD(b_i)$ to node $j$, it cannot change its bid after observing $b_j$. 
\item The node submitting lower bid is recognized as a leader. The other node sends the list of its clients (if any) to the leader. 
\item The loser node informs its clients (if any) of their new leader. The leader adds the loser node (and its clients) to the list of its clients.

\end{enumerate}
The aforementioned process continues among the winners of previous interactions, i.e., remaining leaders. Such a process continues until the number of participants reaches $n$. In this case, the single remaining leader is recognized as the leader of the MDC. Then, the management procedure enters phase II.

It should be noted that nodes have no incentive to change their bid for RD fee after an interaction. This is due to fact that if one node increases its RD fee, none of its clients accept this node as its leader anymore. If such a node decides to increase its RD fee only for new incoming clients, since the amount of RD should be identical for all clients, the TU punishes this node (if this node is selected as the single leader of the MDC) when it wants to receive its credits by reporting its voucher to the TU. In addition, this node has no incentive to decrease its RD fee because it reduces its benefit based on (\ref{equation:xx1}).

\subsection{Phase II}\label{section:phase2}
In the second phase of LORD, the selected leader begins its task and assumes the leadership role for one time-slot. 
Before the end of the current time-slot, the leader should perform an auction among all participants including itself.

%

In order to implement the auction as a simultaneous game, we define the process of leader selection as follows:

\begin{enumerate}
\item The current leader of the MDC computes the hash value of its bid using a message digest function ($MD(.)$), sends it to all clients, and asks them for their bids.
\item Each client sends its bid to the leader regardless of leader's offer because it cannot reveal the bid of leader.
\item Upon receiving the bids from all clients, the current leader sorts nodes based on the values of their bids. The node which has the lowest bid is selected as the leader and the value of its bid is assigned as its RD fee. The current leader sends the name of the selected leader and its RD fee, as well as the list of other nodes including itself, and their bids to all clients. 

Each client that has received such a message, first calculates the hash of the current leader's bid using $MD(.)$. Then, the result is compared with the one offered by the current leader at the beginning of this auction. The leader does not change its bid because it knows that the clients can find out if it does so.
If the values are identical, the result of auction is valid. Therefore, such a node considers the selected leader as the new leader for the next time-slot.

The selected leader considers itself as the leader of the MDC for the next time-slot.
\end{enumerate}

Based on this process, since the current leader has sent the hash value of its bid at the first step, it cannot change its bid after receiving the values offered by clients. If it does so, all clients will notice and will put the name of such a leader in their blacklists. As a punishment, such a node is deprived of further participation in the MDC. If we consider a MDC with \textit{n} nodes, the total number of messages required in the leader selection is equal to $3(n-1)$, among which $2(n-1)$ number of messages is related to the leader. 

Note that the current leader can misbehave in the leader selection process in two ways. First, it can manipulate the values offered by others in order to select itself as the next leader. Second, it can deviate from performing the auction for selecting a new leader. In order to address these issues, suitable punishment policy such as exclusion from the MDC can be used in order to make sure that the leader do not misbehave. For details of punishment policy, see Appendix \ref{punishment}.

\subsection{Entrance or Departure of Nodes}\label{section:entrance}
After introducing our leader selection process, we focus on entrance and departure of nodes. When a node intends to join the MDC by switching its resource sharing feature on, it sends a message to the leader of MDC. When the leader receives such a message, it invites the entering node to wait until the next leader selection process. At the next leader selection process, all nodes modify the number of participating nodes to $n+1$ and the new node can enter the MDC.
%

The departure event may also occur due to a number of reasons such as mobility and battery depletion. If the departing node is not the leader, the MDC continues its regular function. The leader will remove the name of the departed node from the list of clients. 
In the case of a departing leader, the solution is to start from phase I and perform a leader selection in order to select a new leader for the MDC.

\section{Transaction Protocol}
\label{trans_protocol}

According to LORD, for each RD, the RD requesting node should give the leader a voucher of value RD fee. However, without a secure transaction protocol, an RD requesting node can cheat by not sending the voucher to the leader when it receives RD from the leader. As a result, a protocol is required to guarantee that no party can cheat during the RD procedure. This issue has been addressed well in the literature of fair-exchange protocols (e.g., \cite{x500}, \cite{x502}, \cite{x25}), especially for wireless environments. Regardless of the differences among such protocols, there are two involved parties (called Alice and Bob), and one entity as a mediator (Trusted Third Party (TTP)) in ``optimistic" class of these protocols. Although not repeated here, such protocols guarantee the fair exchange between involved parties. By the assistance of these secure protocols, we present the transaction protocol among the leader, the RD requesting node, and the resource provider as follows.

When an RD requesting node sends a request to the leader to find a resource provider for $RT_t$ (Resource Type $t$), the RD requesting node, the leader, and the resource provider execute the following protocol collectively: 
\begin{enumerate} 
\item Resource requesting node sends an RD request for $RT_{t}$ to the leader. 
\item When the leader receives such a message, it searches its database for a proper resource provider for the specified resource. If the leader cannot find a proper node for that type of resource, it sends RD failure message to the RD requesting node. On the other hand, if the leader finds a resource provider, it sends a message verified by resource provider to the RD requesting node indicating that the proposed resource provider can provide $RT_{t}$. The fact that the leader should send a message confirmed by a resource provider to the RD requesting node prevents the leader from cheating and introducing a fake resource provider to the RD requesting node.

\item When one node (as a resource provider) receives such a message, if it can provide the requested resource, it sends a confirmation message to the leader.
\item Upon receiving the verification from the resource provider, the leader should introduce the resource provider to the RD requesting node and in return receive the voucher from this node. Considering the leader, the RD requesting node, and the TU in our framework as Alice, Bob, and the TTP in the optimistic fair-exchange protocols, respectively, such protocols can ensure a secure transaction between the leader and the RD requesting node. 

\end{enumerate}

Considering the proposed protocol, neither the leader nor the RD requesting node can cheat. Hence, in practice, no node has incentive to cheat and there is no need for the TU to mediate in such situations.

\section{Performance Evaluation} 
%
%
%
%
\subsection{LORD Performance}
We compare LORD with decentralized model for RD. To this end, we consider a MDC including 10 nodes \footnote[1]{As our analytical analysis (Propositions 1 and 2) indicates LORD works for large number of nodes. Limiting the number of nodes to 10 is only for the sake of clarity of results presentation.} where each one is randomly assigned an initial energy level in the interval $(0,100)$ joles. Then, based on the initial energy levels, modified RDPC of each node is calculated in the interval $(0, 1)$ joles (the higher the initial energy level, the lower the modified RDPC). For example, if the initial energy level of node $i$ is 60 joles, the modified RDPC of this node is 0.4 joles. Then, considering the relationship between messaging and processing costs of node $i$ as $c_i ^M = \theta c_i ^P$, $c_i ^P$ and $c_i ^M$ can be calculated using (\ref{equation:xx2232}). By substituting the value of $c_i ^M$ in (\ref{equation:xx2233}), modified RDC of node $i$ can be computed. 
We assume $T_{select}$ is equal to 2 minutes in LORD, and the leader provides 3 RDs for each node at every time-slot (i.e., $\eta=3$). Similarly, we assume that each node performs 3 RDs itself every 2 minutes in decentralized model for RD. Furthermore, for simplicity, we assume $m_i = \eta =3 \hspace{2mm}$ for all $i$. 

In order to compare LORD with decentralized model, we consider three general cases, $c_i ^{M} = c_i ^{P}$, $c_i ^{M} = 10c_i ^{P}$, and $c_i ^{M} = 0.1c_i ^{P}$.

\begin{enumerate}
\item Fig. \ref{figure7} compares LORD with the decentralized model when $c_i ^{M} = c_i ^{P}$. In Fig. \ref{figure7.a} and Fig. \ref{figure7.b}, we have plotted the percentage of nodes' energy over time. Fig. \ref{figure7.a} indicates that LORD is able to balance the energy consumption among all nodes. However, the energy consumption of nodes is unbalanced in the decentralized RD model (Fig. \ref{figure7.b}), leading to earlier death of some nodes. Fig. \ref{figure7.c} compares the percentage of alive nodes in LORD with that of the decentralized model. We consider
a node to be alive as long as it has not depleted its energy.
This figure demonstrates that overall lifetime of MDC participants is higher in LORD. In Fig. \ref{figure7.d}, we compare the improvement in payoff of the participating nodes in LORD in comparison with that of in the decentralized model for RD. This figure represents that participating nodes always have higher payoffs in LORD. Furthermore, if we sort nodes based on their energy levels after 10 minutes ($EL$) from Fig. \ref{figure7.a}, we have $EL_3>EL_7>EL_8>EL_{10}>EL_1>EL_2>EL_9>EL_{5}>EL_4>EL_6$; and if we sort nodes based on their improved payoff after 10 minutes ($\Delta U$) from Fig. \ref{figure7.d}, we have $\Delta U_3<\Delta U_7<\Delta U_8<\Delta U_{10}<\Delta U_1<\Delta U_2<\Delta U_9<\Delta U_{5}<\Delta U_4<\Delta U_6$. This indicates that LORD is more lucrative for nodes with lower levels of energy.

\begin{figure}
    \centering
  \subfigure[Energy level of nodes in LORD]
  {
       \includegraphics[width=1.62in,keepaspectratio]{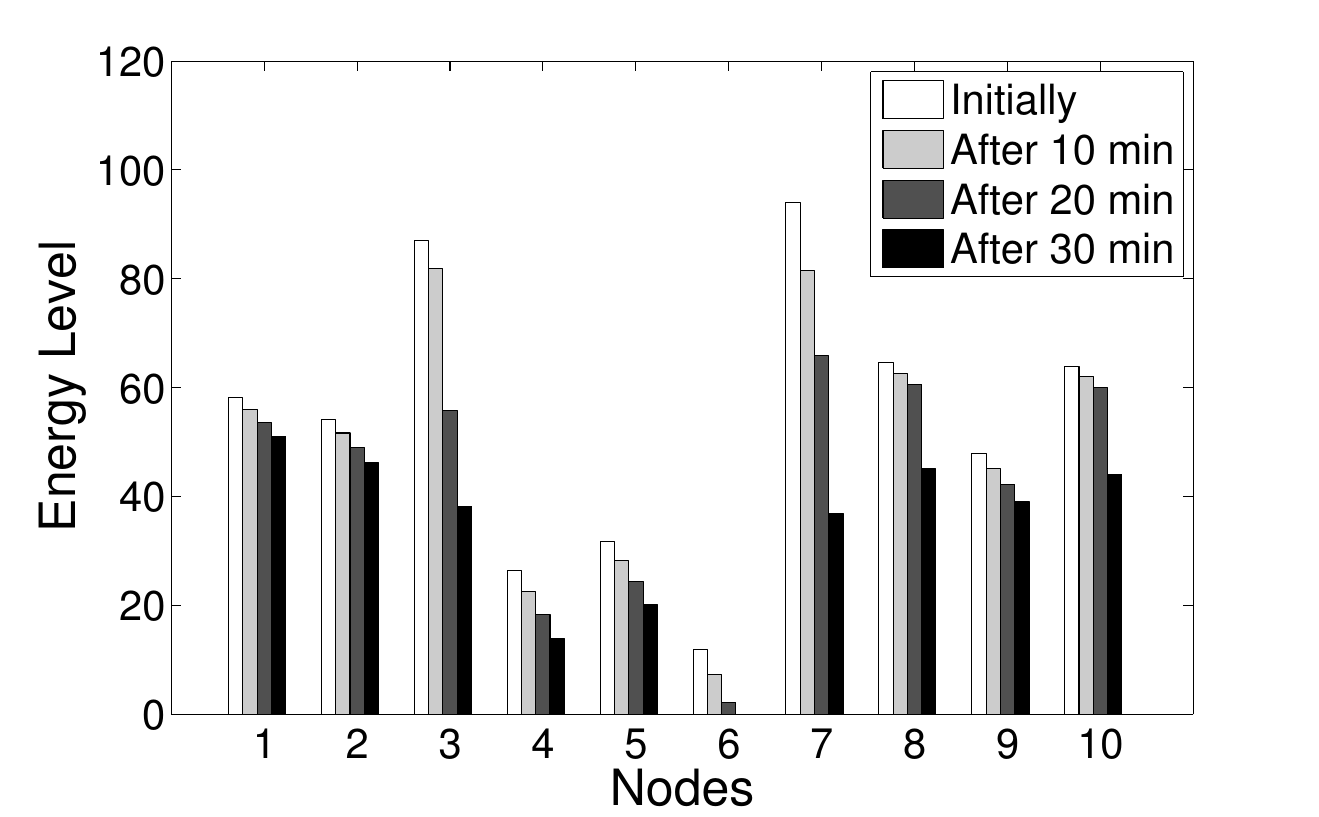}
      \label{figure7.a}
    }
   \subfigure[Energy level of nodes in decentralized model for RD]
    {
        \includegraphics[width=1.62in,keepaspectratio]{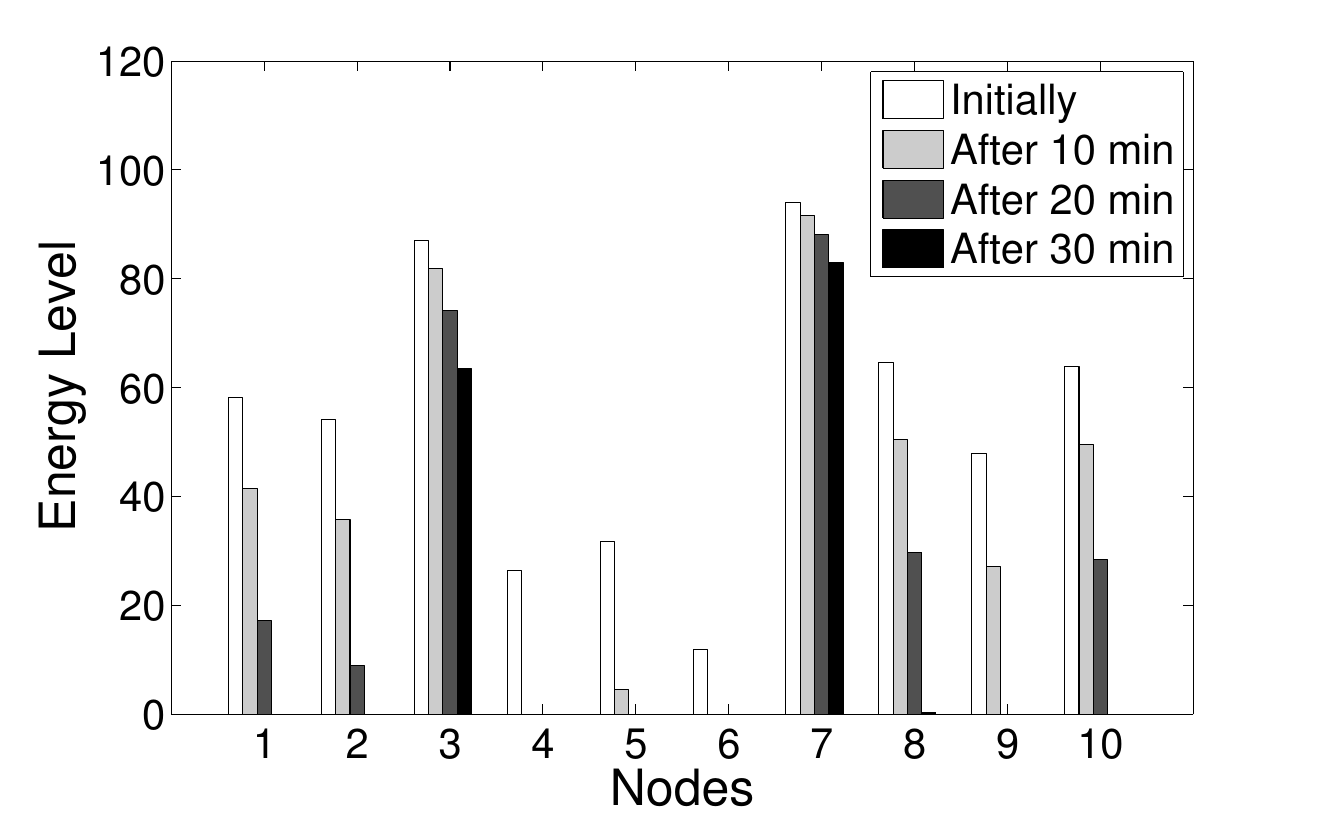}
      \label{figure7.b}
   }
      \subfigure[Percentage of alive nodes]
    {
        \includegraphics[width=1.62in,keepaspectratio]{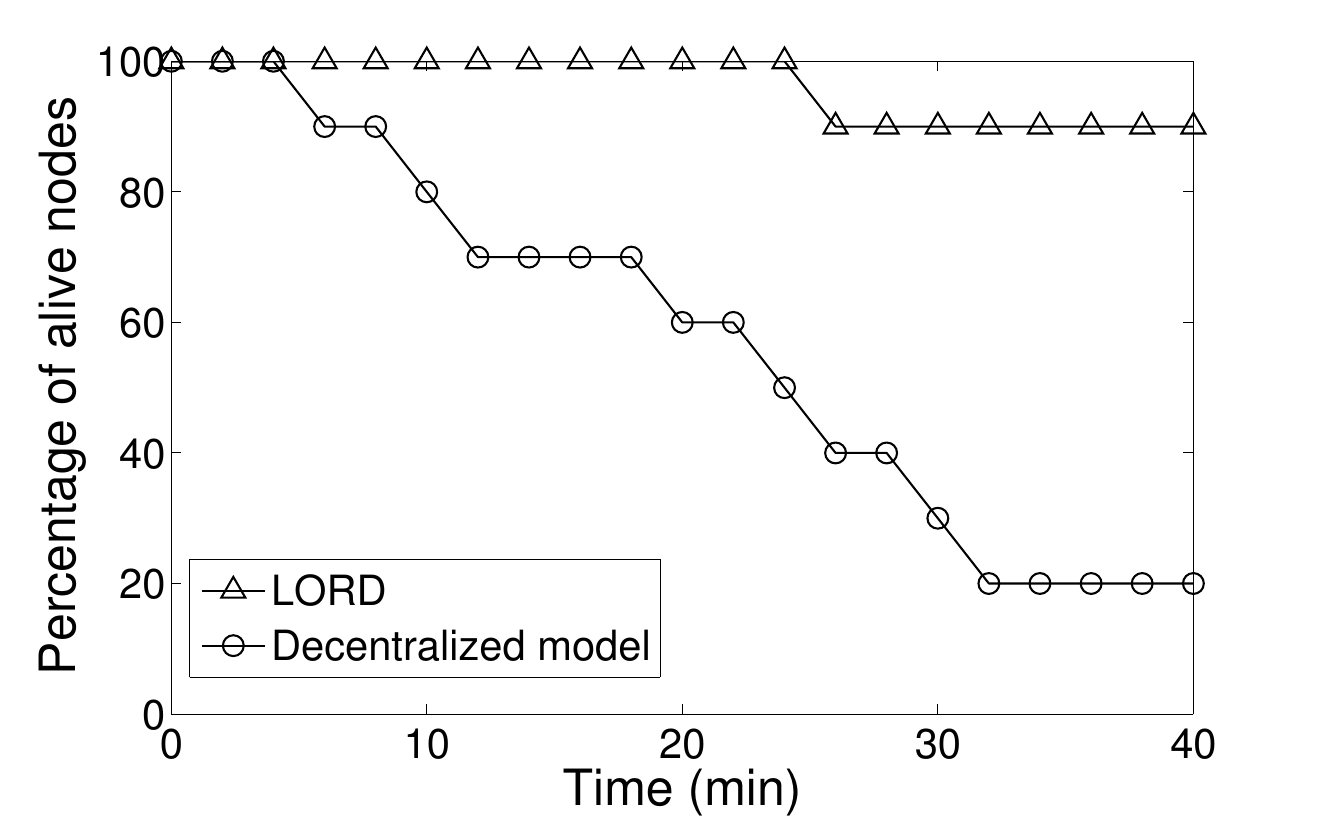}
      \label{figure7.c}
   }
   \subfigure[Improved payoff for LORD in comparison with decentralized model]
    {
        \includegraphics[width=1.62in,keepaspectratio]{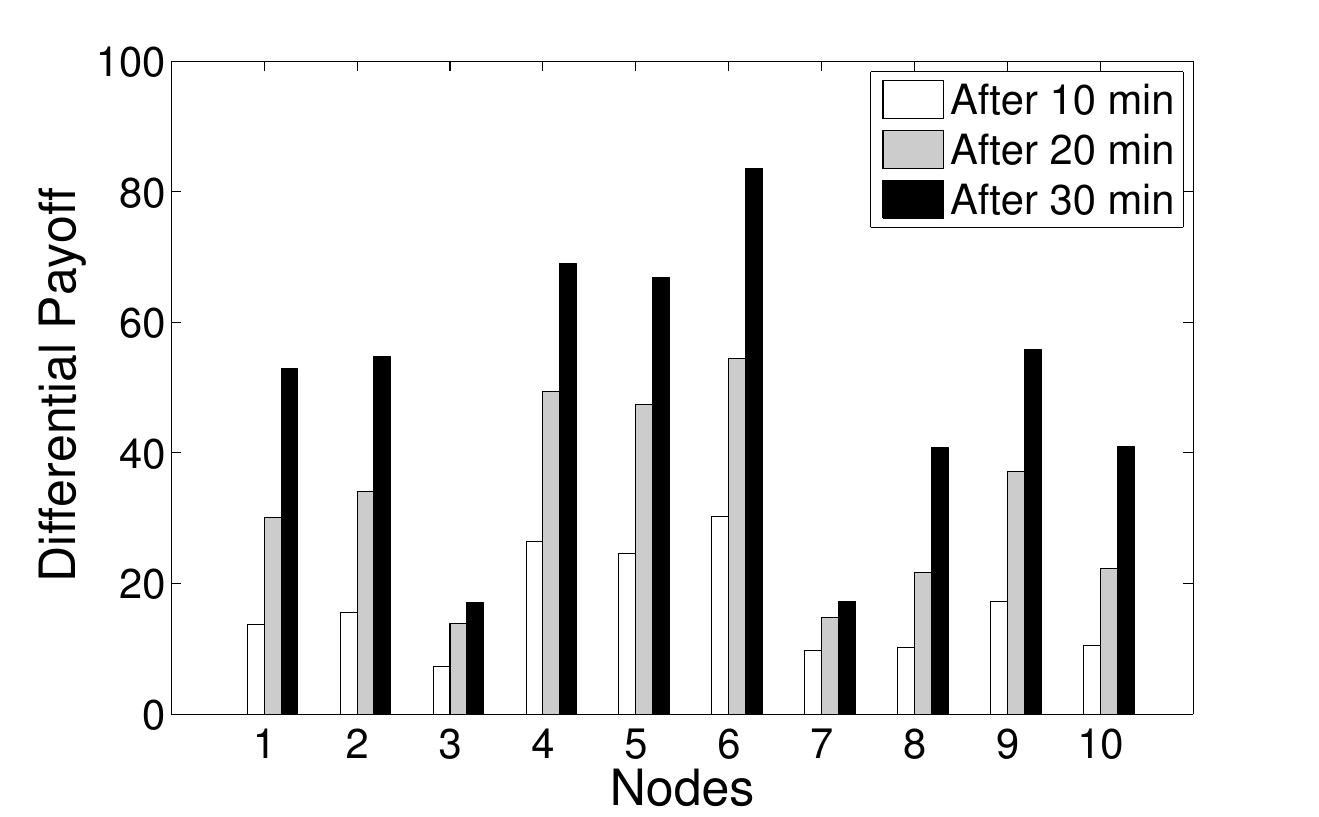}
      \label{figure7.d}
   }
 \caption{Comparison of LORD with the decentralized model for $\theta = 1$, i.e., $c_i ^{M} = c_i ^{P}$}
\label{figure7}
\end{figure}

\item Fig. \ref{figure8} compares LORD with the decentralized model for $c_i ^{M} = 10c_i ^{P}$. Similar to case $c_i ^{M} = c_i ^{P}$, Fig. \ref{figure8.a} and Fig. \ref{figure8.b} indicate that LORD is able to balance the energy consumption among all nodes while decentralized RD model is unable to do so. Comparison of Fig. \ref{figure8.c} with Fig. \ref{figure7.c}, and Fig. \ref{figure8.d} with Fig. \ref{figure7.d} demonstrates that the overall lifetime and improvement in payoff of the participating nodes are higher in LORD than these values in the decentralized model when $c_i ^{M} = 10c_i ^{P}$.

\begin{figure}
    \centering
  \subfigure[Energy level of nodes in LORD]
  {
       \includegraphics[width=1.62in,keepaspectratio]{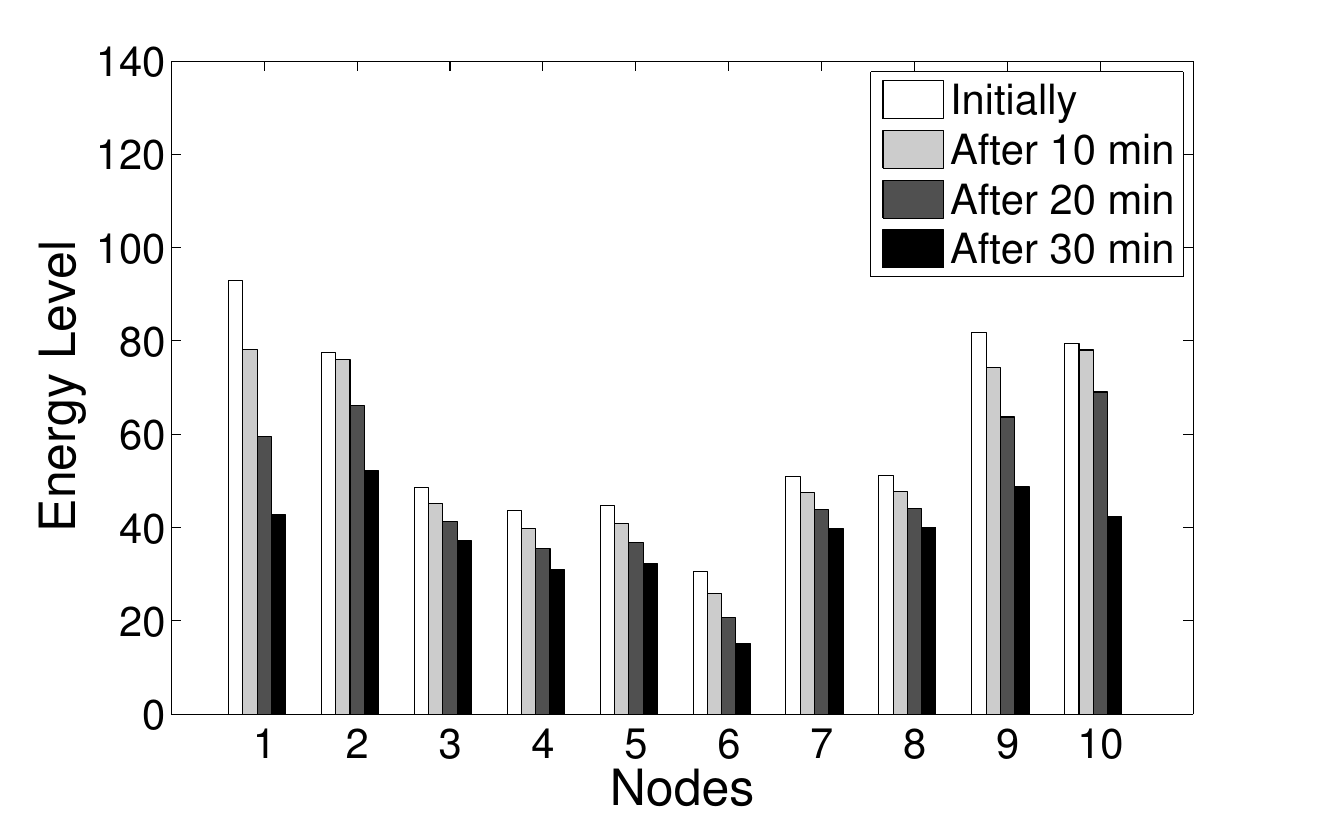}
      \label{figure8.a}
    }
   \subfigure[Energy level of nodes in decentralized model for RD]
    {
        \includegraphics[width=1.62in,keepaspectratio]{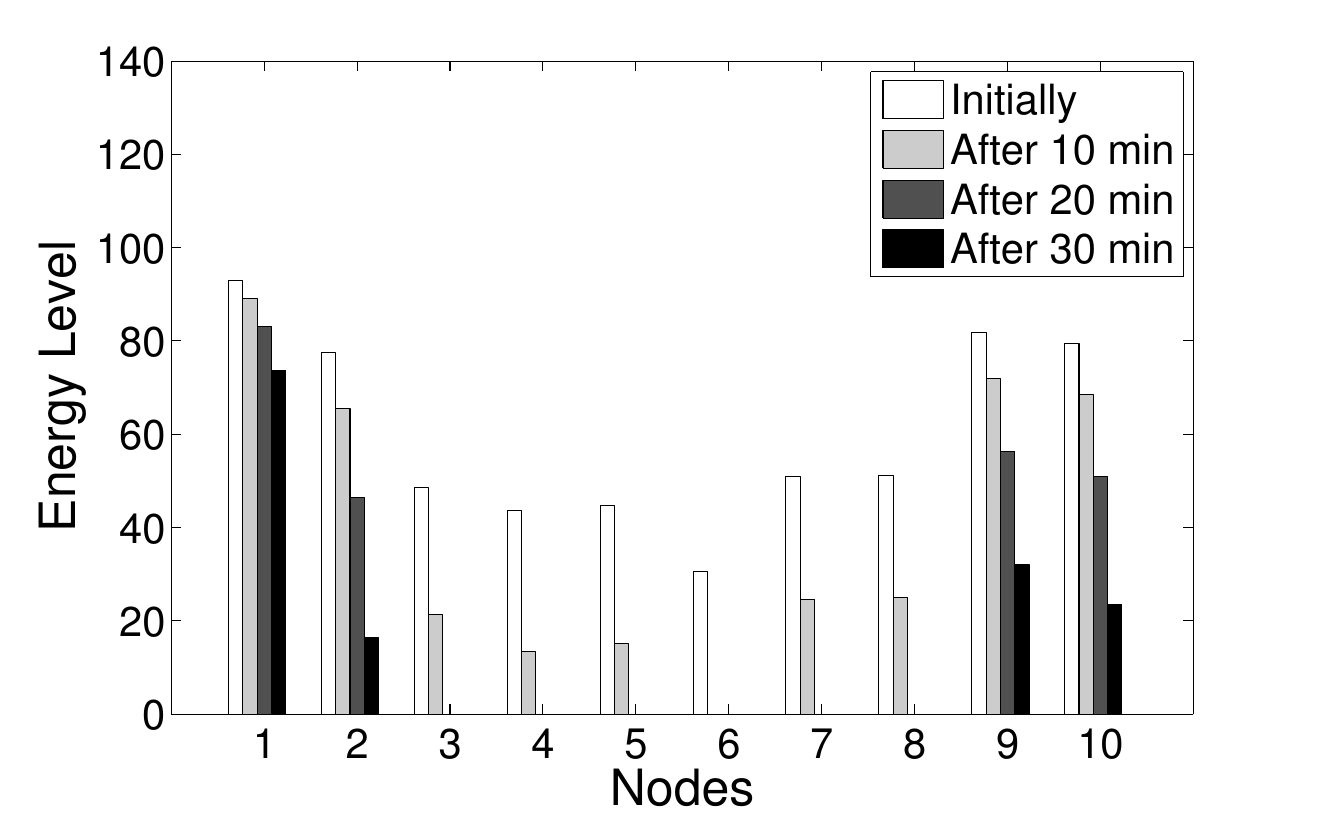}
      \label{figure8.b}
   }
      \subfigure[Percentage of alive nodes]
    {
        \includegraphics[width=1.62in,keepaspectratio]{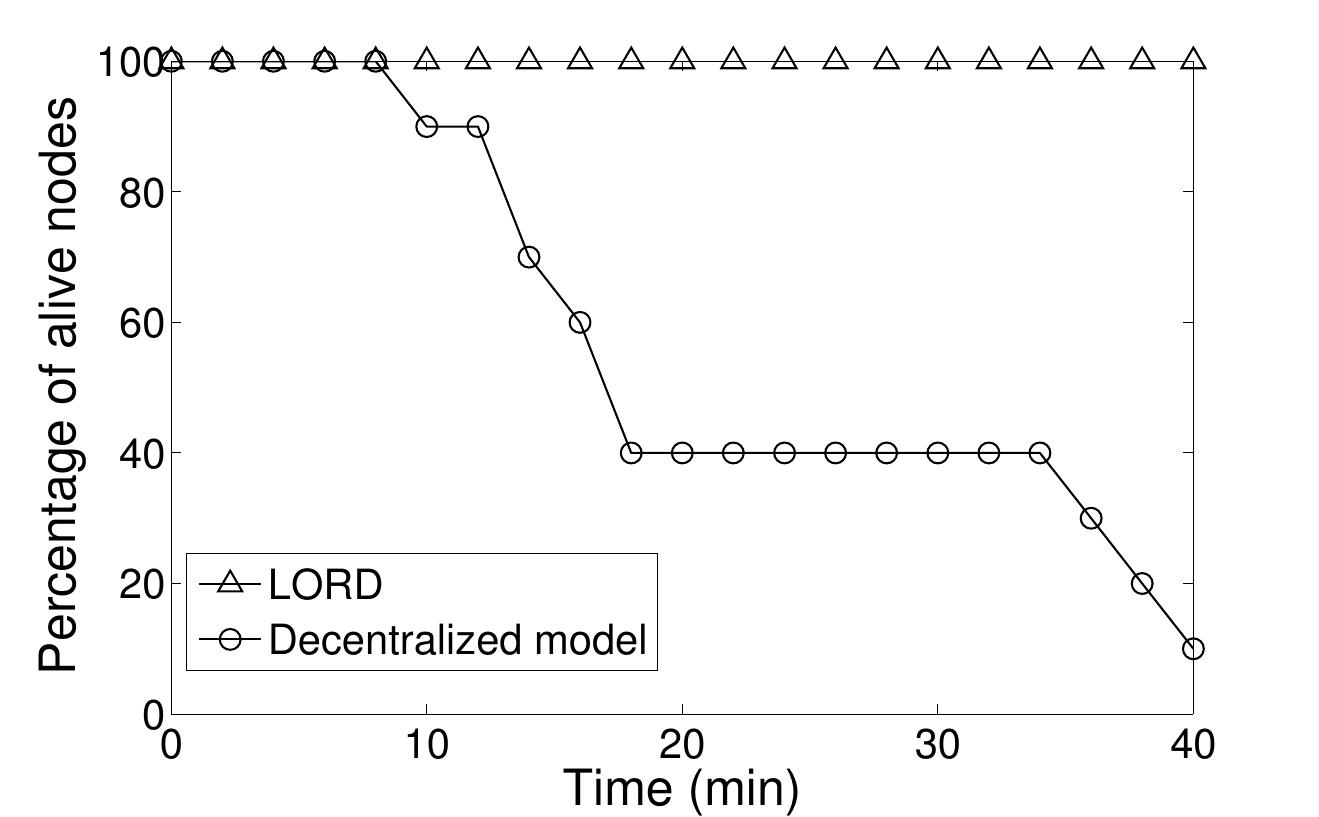}
      \label{figure8.c}
   }
   \subfigure[Improved payoff for LORD in comparison with decentralized model]
    {
        \includegraphics[width=1.62in,keepaspectratio]{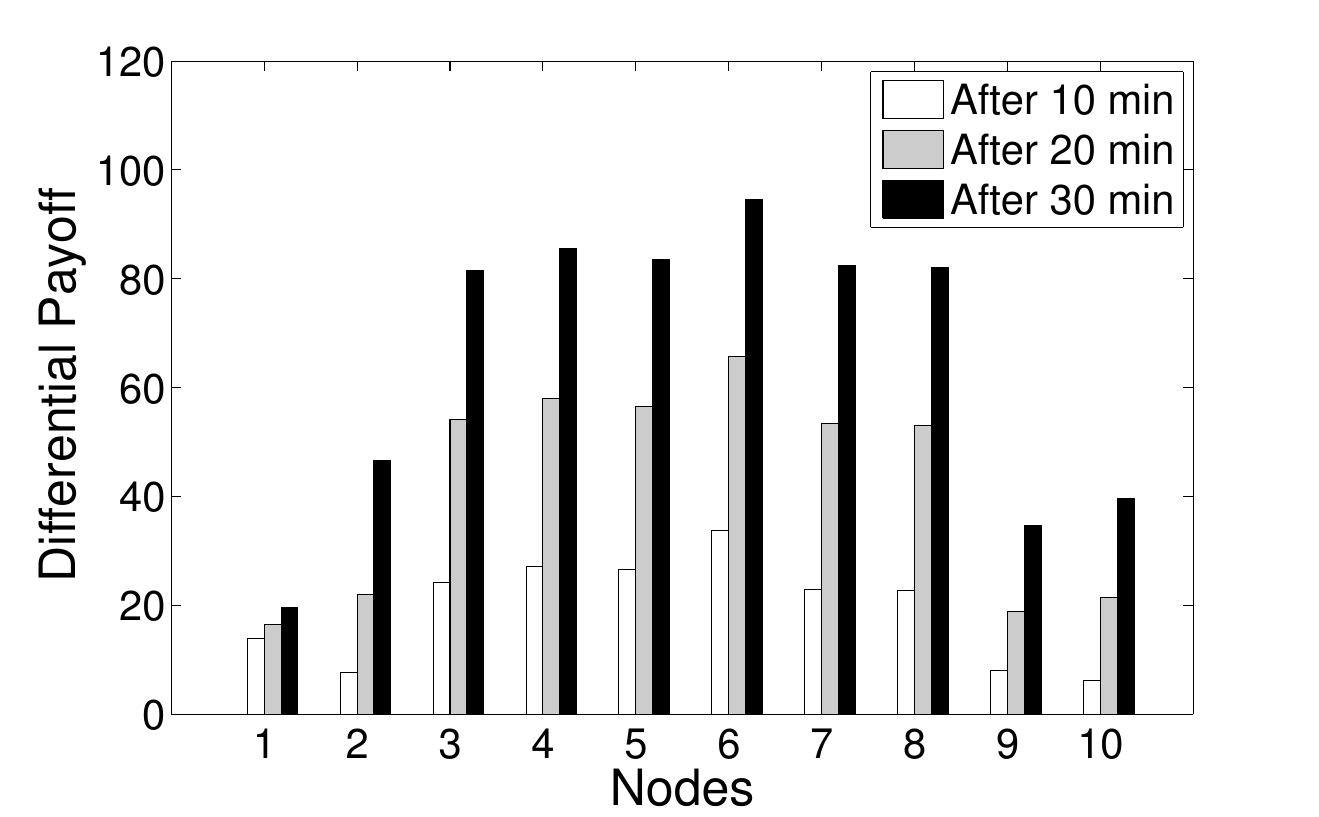}
      \label{figure8.d}
   }
 \caption{Comparison of LORD with the decentralized model for $\theta = 10$, i.e., $c_i ^{M} = 10c_i ^{P}$}
\label{figure8}
\end{figure}

\item Fig. \ref{figure9} compares LORD with the decentralized model for $c_i ^{M} = 0.1c_i ^{P}$. In this case, messaging cost is negligible in comparison with processing cost. Since RD in decentralized model imposes only messaging cost while RD in LORD imposes both processing and messaging costs, Fig. \ref{figure9.a}, Fig. \ref{figure9.b}, and Fig. \ref{figure9.c} indicate that nodes die earlier in LORD in comparison with decentralized RD model because RD in LORD imposes more cost on nodes. Furthermore, Fig. \ref{figure9.d} indicates that the payoff of nodes is lower in LORD in comparison with the payoff of nodes in decentralized RD model.

\begin{figure}
    \centering
  \subfigure[Energy level of nodes in LORD]
  {
       \includegraphics[width=1.62in,keepaspectratio]{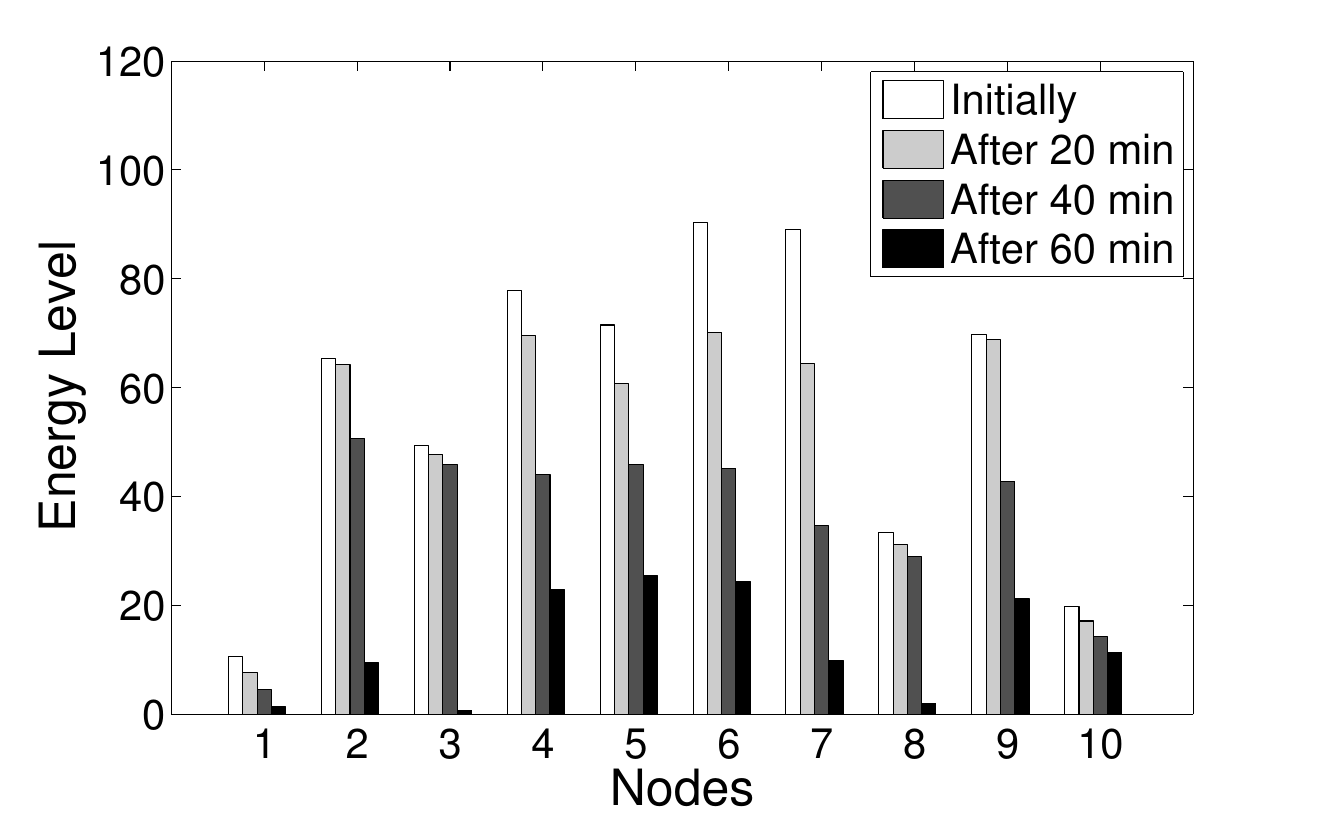}
      \label{figure9.a}
    }
   \subfigure[Energy level of nodes in decentralized model for RD]
    {
        \includegraphics[width=1.62in,keepaspectratio]{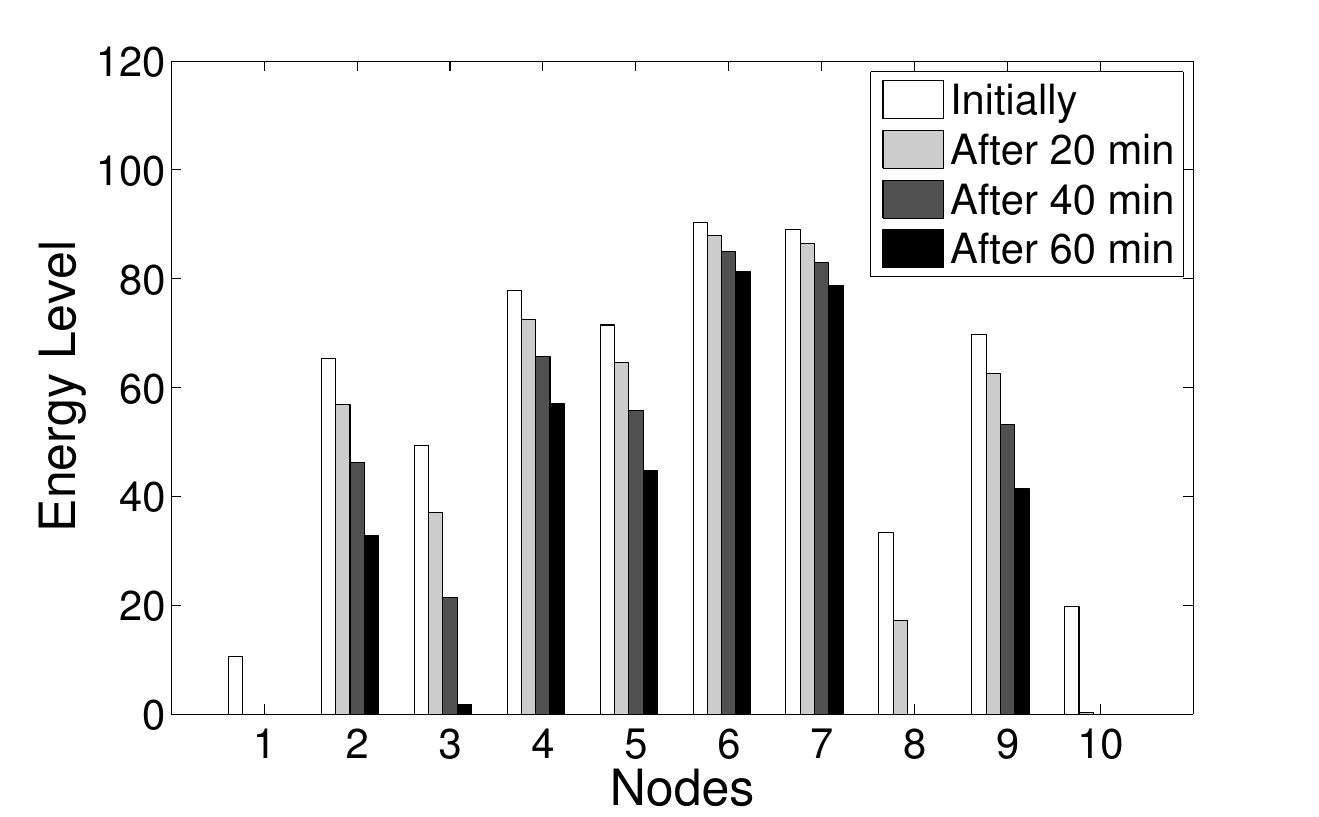}
      \label{figure9.b}
   }
      \subfigure[Percentage of alive nodes]
    {
        \includegraphics[width=1.62in,keepaspectratio]{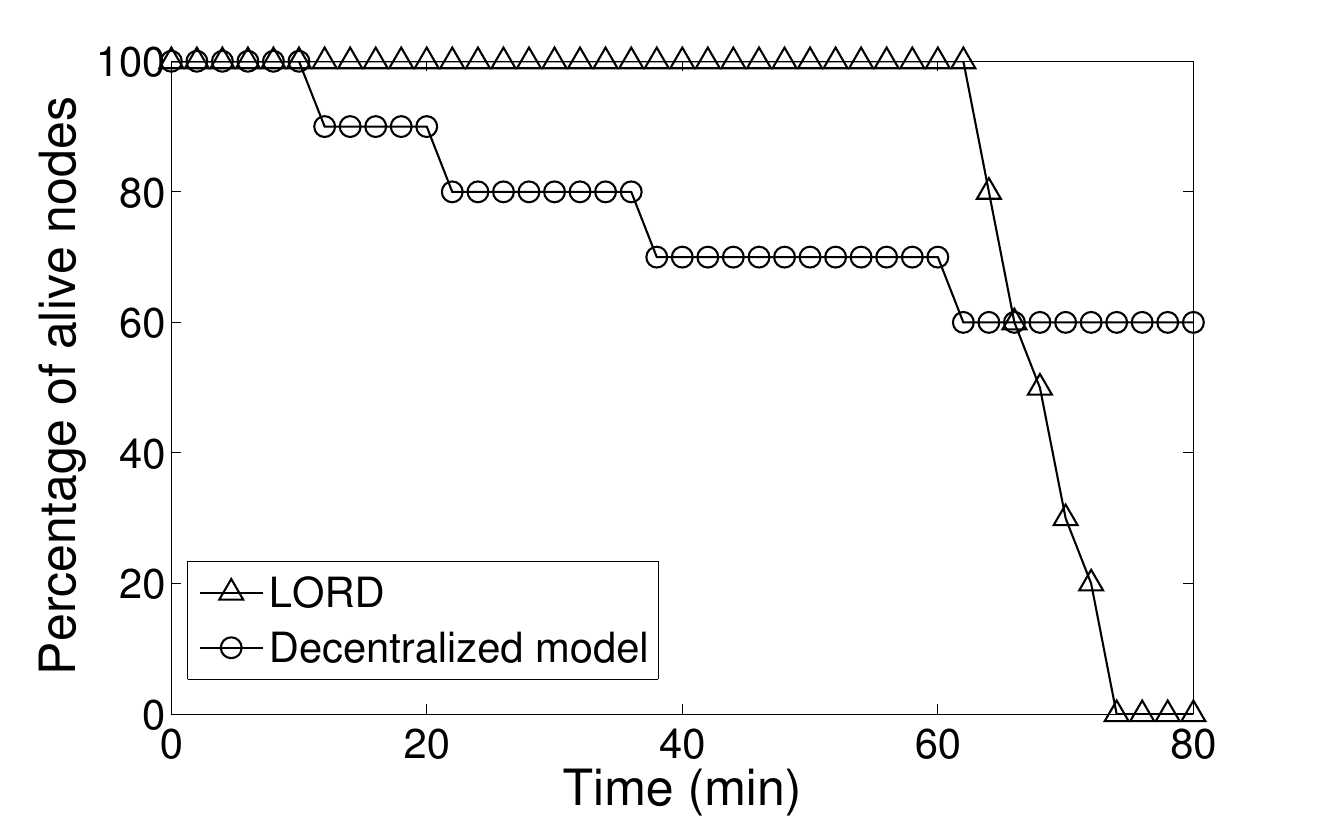}
      \label{figure9.c}
   }
   \subfigure[Improved payoff for LORD in comparison with decentralized model]
    {
        \includegraphics[width=1.62in,keepaspectratio]{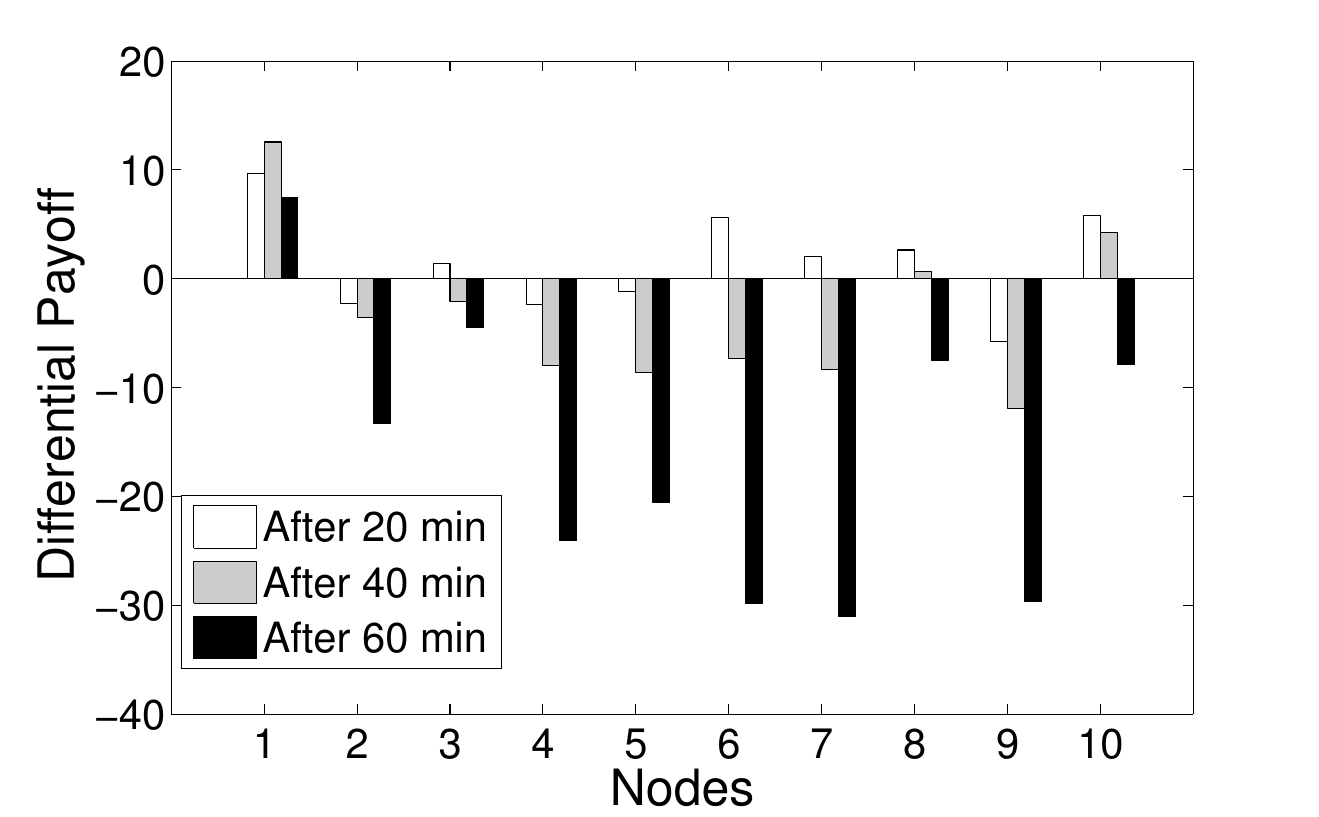}
      \label{figure9.d}
   }
 \caption{Comparison of LORD with the decentralized model for $\theta = 0.1$, i.e., $c_i ^{M} = 0.1c_i ^{P}$}
\label{figure9}
\end{figure}

\end{enumerate}

Considering the aforementioned cases, if $c_i ^{M} \geq c_i ^{P}$, we can conclude that,
\begin{enumerate}
\item LORD ascertains our goal in balancing the energy consumption among participating nodes while the energy consumption of nodes is unbalanced in decentralized RD model.
\item LORD increases the overall lifetime of nodes in comparison with the decentralized RD model. Since in a MDC nodes provide resource for each other, as the percentage of alive nodes is higher, the probability of finding a proper resource provider increases. However, death of several nodes due to unbalanced energy consumption in the decentralized model decreases the probability of performing a successful RD.
\item Nodes prefer participation in MDCs where RD is based on LORD to the ones where RD is carried out by each node individually (i.e., decentralized RD model) since they have higher payoffs in LORD. Furthermore, the improvement in payoff of nodes increases over time since as the energy level of nodes decrease, performing RD by nodes themselves imposes much more cost on them. In addition, LORD is even more lucrative for nodes with lower levels of energy.

\item As the amount of $c_i ^{M}$ increases in comparison with $c_i ^{P}$, the advantage of LORD is more significant in comparison with decentralized RD model.  
\end{enumerate}

On the other hand, since the messaging cost is negligible when $c_i ^{M} < c_i ^{P}$, it is better for nodes to perform RD themselves instead of participating in a MDC where RD is based on LORD.

%
\subsection{LORD Communication Overhead Analysis}
We evaluate the communication overhead imposed on the network in the leader selection algorithms proposed for phases I and II. 
\subsubsection{Overhead Analysis: Phase I} \label{Performance Overhead_1}
The leader selection process proposed for Phase I consists of a series of two-node interactions. At each interaction, one node (loser of the interaction) is added to the list of clients. Therefore, if we consider a MDC with $n$ nodes, the number of required interactions is $n-1$. Each proposed interaction has five steps as explained in Section \ref{section:Phase1}. 
Each of the first four steps requires one message and therefore, the total number of messages of the first four steps in all $n-1$ interactions is $4(n-1)$. However, the number of messages of step 5 depends on the number of clients of the loser node which itself depends on the arrangement of earlier two-node interactions. If we denote the total number of messages of last step in all $n-1$ interactions by $n_l$, total number of messages required in the leader selection procedure is equal to $4(n-1)+ n_l$. Although $n_l$ cannot be calculated exactly, its average value can be approximated by performing enough number of simulations as indicated in Fig. \ref{figure 10}.

\subsubsection{Overhead Analysis: Phase II}
The leader selection process proposed for Phase II consists of one multi-player auction which has three steps as explained in Section \ref{section:phase2} where each step requires $n-1$ number of messages. As a result, if we consider a MDC with \textit{n} nodes, the total number of messages required in the leader selection is equal to $3(n-1)$.
\subsubsection{Overhead Analysis: Comparison of Phase I and Phase II}
Fig. \ref{figure 10} indicates the imposed communication overhead for different number of nodes. It can be seen that the number of messages required in the leader selection of phase I and II are $5.9(n-3.25)$ and $3(n-1)$, respectively. Therefore, the results of simulation indicate that the total number of messages required for the leader selection in both phases is in the order of $n$.

\begin{figure}
\centering
      \includegraphics[width=2.1in,keepaspectratio]{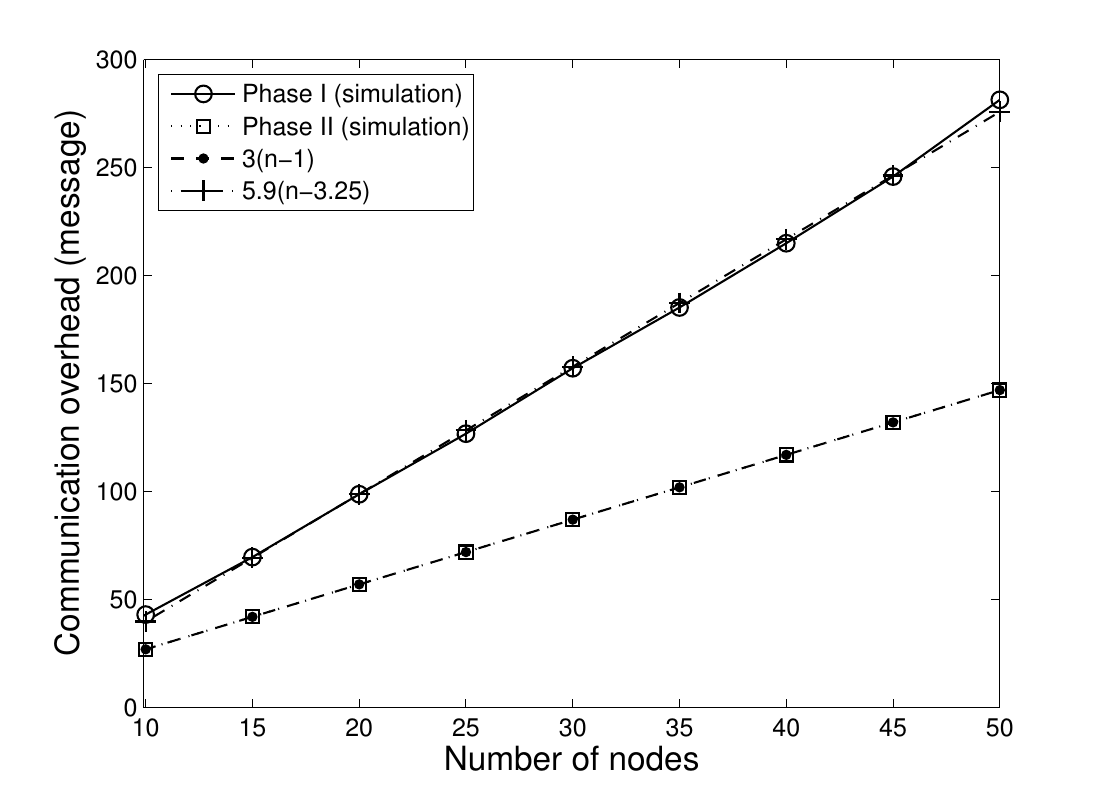}
      \caption{Communication overhead imposed on the network}
\label{figure 10}
\end{figure}
\section{Conclusion}
In this paper, we have provided a new solution for Resource Discovery (RD) in mobile device clouds (MDCs). In order to balance the energy consumption of all nodes and prolong the overall lifetime of MDCs, we have proposed LORD, a leader-based framework for RD in MDCs, in which the RD task is handled by a node selected as the leader. We have proposed a multi-player first-price sealed-bid auction to select the most qualified node as the leader. 
Through simulation results, we have demonstrated that LORD is able to balance the energy consumption among mobile devices and increase their payoffs in comparison with the decentralized RD model.

\newpage
\appendices
\section{Validating the Assumption $c_i < \overline{c}_i$}
\label{validating}
Let's consider three general cases, $c_i ^{M} = 10c_i ^{P}$, $c_i ^{M} = c_i ^{P}$, and $c_i ^{M} = 0.1c_i ^{P}$. In order to be able to compare RDPC and RDC, we assume that number of nodes is at least 10 ($n\geq 10$), number of RDs for each node is at least 1 ($\eta \geq 1$), and the number of times the leader searches the database for finding a resource provider for itself is smaller than or equal to 10 ($m_i \leq 10$). Considering these assumption as well as (\ref{equation:xx333}), it can be easily shown that:
\begin{align}
\label{equation:xx666}
&1 + \dfrac{m_i +1}{M}< 2  \hspace{3mm} \&  \hspace{2mm} 3+\dfrac{2}{\eta} < 5 \hspace{5mm} \Longrightarrow \nonumber \\
&c_i = (1 + \dfrac{m_i +1}{M})c_i ^{P} + (3+\dfrac{2}{\eta})c_i ^{M}  < 2c_i ^{P} + 5c_i ^{M}
\end{align} 
Furthermore, considering aforementioned assumptions and (\ref{equation:xx5}), it can be shown that:
\begin{align}
\label{equation:xx667}
 \overline{c}_i = (n-1)c_i^M \geq 9c_i ^{M}
\end{align} 
Now, we compare (\ref{equation:xx666}) and (\ref{equation:xx667}) in three aforementioned cases,
\begin{itemize}
\item $c_i ^{M} = 10c_i ^{P}$. In this case, we have $c_i < 52c_i ^{P}$ and  $\overline{c}_i \geq 90 c_i ^{P}$, and therefore we can conclude that $c_i < \overline{c}_i$. 
\item $c_i ^{M} = c_i ^{P}$. In this case, we have $c_i < 7c_i ^{P}$ and  $\overline{c}_i \geq 9 c_i ^{P}$, and therefore we can conclude that $c_i < \overline{c}_i$.
\item $c_i ^{M} = 0.1c_i ^{P}$. In this case, we have $c_i < 2.5c_i ^{P}$ and  $\overline{c}_i \geq 0.9 c_i ^{P}$, and therefore we cannot conclude that $c_i < \overline{c}_i$.
\end{itemize}

Considering these cases, if $c_i ^{P} \leq c_i ^{M}$ (which is plausible since $c_i ^{P}$ is a negligible cost due to high processing power available in today's smartphones \cite{x29}), our assumption that $c_i < \overline{c}_i$ is valid.

\section{Proof of Lemma 1}
\label{p_l1}
Rewriting (\ref{equation:x13}) as follows:
\begin{equation}
\begin{aligned}
\label{equation:x15}
\sigma^{'}(\widehat{c}) - \dfrac{(n-1)M+\eta}{M(K-\widehat{c})}\sigma(\widehat{c}) = - \dfrac{\widehat{c}(n-1)}{K-\widehat{c}}
\end{aligned}
\end{equation}
we get to the form of first order linear differential equation that is $y^{'} (x) + p(x)y = q(x)$. Considering the general solution of this equation, the solution of (\ref{equation:x15}) can be written as follows:
\begin{align}
\label{equation:x16}
\sigma (\widehat{c}) = u^{-1}(\widehat{c}) \int u(\widehat{c}) \dfrac{\big(-\widehat{c}(n-1)\big)}{K-\widehat{c}} d\widehat{c}
\end{align}
where $u(\widehat{c})= exp(\int \dfrac{-A}{K-\widehat{c}} d\widehat{c})= (K-\widehat{c})^A$ in which $A= \dfrac{(n-1)M+\eta}{M}$. Therefore,
\begin{align}
\label{equation:x17}
\sigma (\widehat{c}) = (K-\widehat{c})^{-A} \int (K-\widehat{c})^{A}(n-1)\big(1- \dfrac{K}{K-\widehat{c}}\big) d\widehat{c}
\end{align}
Simplifying (\ref{equation:x17}), the solution is,
\begin{align}
\label{equation:x18}
\sigma (\widehat{c})= \frac{1}{A+1}(n-1)(\widehat{c} + \frac{K}{A})
\end{align} 
Substituting $A= \dfrac{(n-1)M+\eta}{M}$ where $M=(n-1)\eta$ in (\ref{equation:x18}) results in,
\begin{align}
\label{equation:x140}
b = \sigma (\widehat{c})= \frac{(n-1)^2}{n(n-1)+1}\Big(\widehat{c} + \frac{K(n-1)}{(n-1)^2+1}\Big)
\end{align} 

\section{Proof of Proposition 1}
\label{p_p1}
In order to have $b_i < \overline{c}_i \hspace{3mm} \forall i$, based on (\ref{equation:x14}) we need,
\begin{align}
b_i = \frac{(n-1)^2}{n(n-1)+1}\Big(\widehat{c}_i + \frac{K(n-1)}{(n-1)^2+1}\Big) \leq \overline{c}_i 
\label{eq1} 
\end{align} 
Substituting $\widehat{c}_i$ from (\ref{equation:xx2232}) and $\overline{c}_i$ from (\ref{equation:xx5}) in (\ref{eq1}),
\begin{align}
&\frac{(n-1)^2}{n(n-1)+1}\Big((1 + \dfrac{m_i +1}{M})c_i ^{P} + \big(3+\dfrac{2}{\eta} - \dfrac{1}{n-1} - \dfrac{1}{M}\big)c_i ^{M} 
\nonumber
\\
&+ \frac{K(n-1)}{(n-1)^2+1}\Big) \leq (n-1)c_i ^M 
\label{eq3} 
\end{align}

After simplifying, the value of $\eta$ that satisfies (\ref{eq3}) can be obtained as follows:
 \begin{equation}  
\label{eq_2}
\eta \geq \dfrac{\eta_{num}}{\eta_{den}}
\end{equation}
where,
\begin{align*}
\eta_{num} &= (n^2 -2n +2) \Big((m_i +1)c_i ^{P} + (2n-3)c_i ^{M}\Big)
\\
\eta_{den} &= (n^2 -4n+1)(n^2 -2n +2)c_i ^M
\\
& - (n-1)^2 K -(n-1)(n^2 -2n +2)c_i ^{P}
\end{align*}

Recall that $c_i^{P}$ is a function of $n$ and the resource status of node $i$ ($RS_i$), $c_i^{M}$ is a function of only $RS_i$. RHS of (\ref{eq_2}) is a function of $n$, $RS_i$, and $m_i$, therefore (\ref{eq_2}) can be rewritten as, 
\begin{align} 
\label{eq_4}
\eta &\geq \eta_{min} (n,RS_i,m_i)  
\end{align} 

Let us denote the upper bound of $m_i$ for all $i$ by $\lambda$. If we assume that the resource status of all nodes to belong to $[RS_{min}, RS_{max}]$, the minimum number of RD queries done by the leader for each client in a mobile device cloud (MDC) of size $n$ to satisfy $b_i < \overline{c}_i \hspace{2mm} \forall i$ is,
\begin{align}
\label{contraint_1}
\eta_{min}(n)=\max_{RS_i\in [RS_{min}, RS_{max}]}\eta_{min}(n,RS_i,\lambda) 
\end{align}
Choosing $\eta=\eta_{min}(n)$ in a MDC of size $n$ is a sufficient condition to ensure  $b_i < \overline{c}_i \hspace{3mm} \forall i$.  

To see whether (\ref{contraint_1}) is a hard constraint or not, we analyze this constraint numerically. To this end, we consider a MDC of varying size $n$. For each $n$, we assume RDPC (as a representation for resource status) of nodes are distributed uniformly in the interval $(0,1)$. Then, considering the relationship between messaging and processing costs of node $i$ as $c_i ^M = 10 c_i ^P$ and using (\ref{equation:xx2232}), $c_i ^P$ and $c_i ^M$ can be calculated. Furthermore, we assume $\lambda = n$ (that is, a leader performs RD at most $n$ times for itself). We change $n$ from 10 to 30 and for each $n$, we find $\eta_{min}(n)$ according to (\ref{contraint_1}). Figure \ref{eta} shows $\eta_{min}(n)$ for various values of $n$. As can be seen from this figure, for $n>10$, we need $\eta_{min}(n) <1$. Since we know that the leader provides at least 1 RD for each client (that is, $\eta \geq 1$), the requirement on $\eta_{min}$ is satisfied. Hence, there is no constraint on $\eta_{min}$ and $b_i < \bar{c}_i$ holds.

\begin{figure}
\begin{center}
\includegraphics[width=3.3in,keepaspectratio]{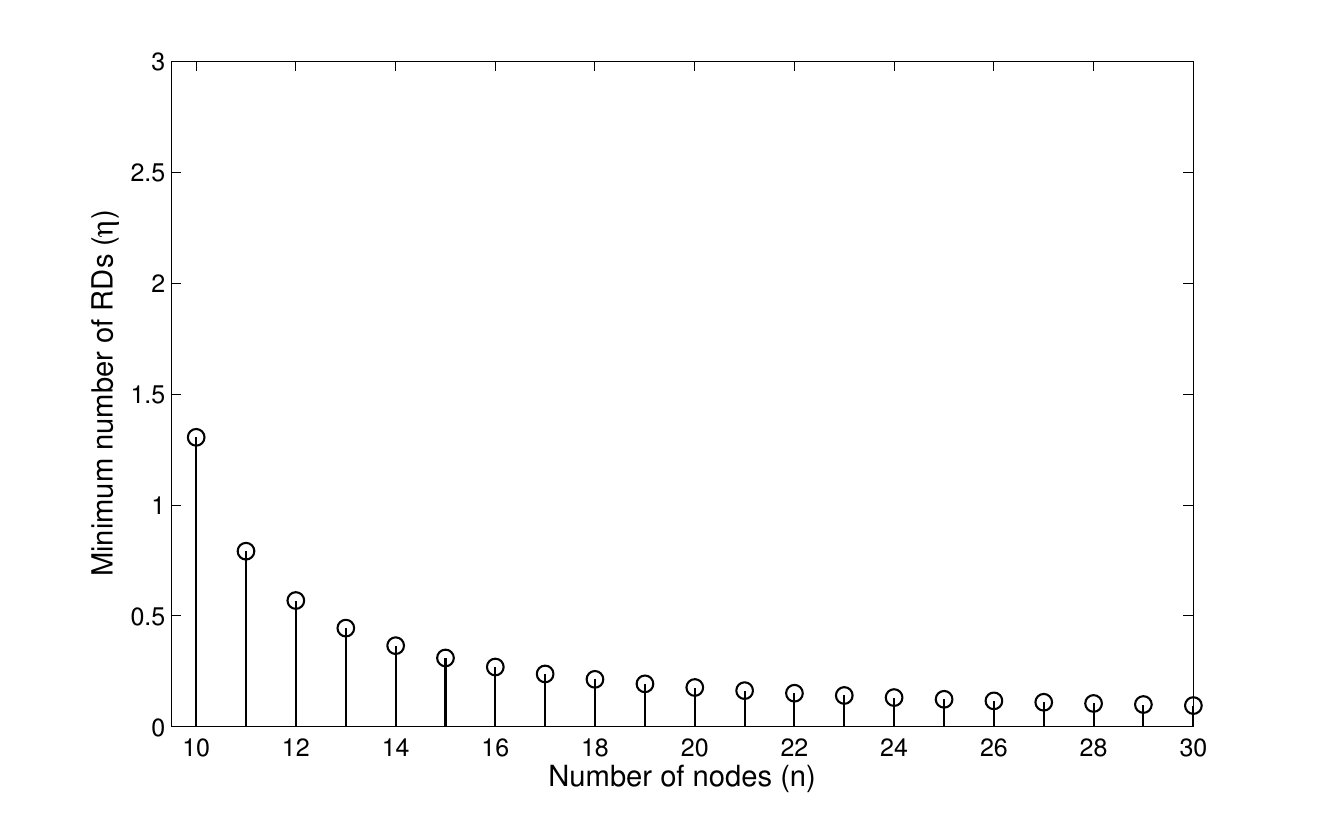}
\end{center}
\caption{Minimum number of RDs ($\eta$) for a MDC with different sizes}
\label{eta}
\end{figure}

\section{Proof of Theorem 1}
\label{p_t1}
If all players except for player $i$ play based on (\ref{equation:x14}), the expected payoff for node $i$ can be written as:

\begin{small}
\begin{align} 
\label{equation:x19}
&\widehat{U}_i \big(\widehat{c}_i , b_i ;\{\sigma(\widehat{c}_k)\}_{k \neq i} \big) = M(b_i - \widehat{c}_i) \big(\frac{K - \sigma ^{-1}(b_i)}{K}\big)^{n-1} 
\nonumber \\
&-
\int_{0} ^{\sigma ^{-1}(b_i)}\frac{\eta(n-1)M}{K(nM+\eta)}\big(\widehat{c}_j + \frac{MK}{(n-1)M+\eta}\big)\big(\frac{K- \widehat{c}_j}{K}\big)^{n-2}d\widehat{c}_j 
\end{align}
\end{small}
Simplifying,
\begin{small}
\begin{align} 
\label{equation:x20}
\widehat{U}_i \big(\widehat{c}_i , b_i ;\{\sigma(\widehat{c}_k)\}_{k \neq i}\big) &= M(b_i - \widehat{c}_i) \big(K - \sigma ^{-1}(b_i)\big)^{n-1}
\nonumber \\
&- \frac{\eta(n-1)M}{nM+\eta}\times\frac{\big(K - \sigma ^{-1}(b_i)\big)^{n}}{n} 
\nonumber
\\ & + \frac{\eta(n-1)MK}{(n-1)M+\eta}\times\frac{\big(K - \sigma ^{-1}(b_i)\big)^{n-1}}{n-1}
\nonumber
\\ &+ \frac{\eta(n-1)M}{nM+\eta}\times\frac{K^{n}}{n} 
\nonumber
\\ &
- \frac{\eta(n-1)M}{(n-1)M+\eta}\times\frac{K^{n}}{n-1}
\end{align}
\end{small}
For any \begin{small}
$b_i \in [\dfrac{(n-1) M^2 K}{(nM+ \eta)\big((n-1)M+\eta\big)},\dfrac{(n-1)MK}{(n-1)M+ \eta}]$:
\begin{align}
\label{equation:x21}
\sigma ^{-1}(b_i)= \frac{(nM+\eta)b_i}{(n-1)M} - \frac{MK}{(n-1)M+\eta}
\end{align}
\end{small}
Substituting (\ref{equation:x21}) in (\ref{equation:x20}) yields:
\begin{small}
\begin{align} 
\label{equation:x22}
\widehat{U}_i \big(&\widehat{c}_i , b_i;\{\sigma(\widehat{c}_k)\}_{k \neq i}\big) 
\nonumber
\\ &= M(b_i - \widehat{c}_i) (nM+\eta)^{n-1}\Big(\frac{K}{(n-1)M+\eta} -\frac{b_i}{(n-1)M}\Big)^{n-1} 
\nonumber
\\ &
- \eta(n-1)M(nM+\eta)^{n-1}\frac{\big(\frac{K}{(n-1)M+\eta} -\frac{b_i}{(n-1)M}\big)^{n}}{n} 
\nonumber
\\ & 
- \frac{\eta(n-1)MK}{(n-1)M+\eta}(nM+\eta)^{n-1}\frac{\big(\frac{K}{(n-1)M+\eta} -\frac{b_i}{(n-1)M}\big)^{n-1}}{n-1} 
\nonumber
\\ &
+ \frac{\eta(n-1)M}{nM+\eta}\times\frac{K^{n}}{n} - \frac{\eta(n-1)M}{(n-1)M+\eta}\times\frac{K^{n}}{n-1}
\end{align}
\end{small}
For any \begin{small} $b_i \in [\dfrac{(n-1) M^2 K}{(nM+ \eta)\big((n-1)M+\eta\big)},\dfrac{(n-1)MK}{(n-1)M+ \eta}]$:
\begin{align}
\dfrac{\partial^2{U_i (\widehat{c}_i , b_i ;\{\sigma(\widehat{c}_k)\}_{k \neq i})}}{\partial {b_i}^2}<0 \nonumber
\end{align}
\end{small}

Therefore, $\widehat{U}_i$ is strictly concave in \begin{small}
$b_i \in [\dfrac{(n-1) M^2 K}{(nM+ \eta)\big((n-1)M+\eta\big)},\dfrac{(n-1)MK}{(n-1)M+ \eta}]$,
\end{small}
and has a maximum value at $b_i$ identified in (\ref{equation:x14}). It is clear that deviation from such $b_i$ is not profitable. Hence, the strategy identified in (\ref{equation:x14}) is the unique symmetric Bayesian Nash equilibrium.

\section{Proof of Proposition 2}
\label{p_p2}
Let $E_i^c(t)$ and $\bar{E}_i^c(t)$ denote consumed energy of node $i$ at time-slot $t$ for LORD and decentralized model, respectively. According to (\ref{equation:xx2}) and (\ref{equation:xx333}), if node $i$ is the leader at time-slot 
$t$, $E_i^c(t)$ can be written as,
\begin{align}
\label{proof_p2_1}
 E_i^c(t)=Mc_i = M \Big( (1 + \dfrac{m_i +1}{M})c_i ^{P} + (3+\dfrac{2}{\eta})c_i ^{M}  \Big)  \end{align}

If node $i$ is a client at time-slot $t$, according to (\ref{equation:xx4}), $E_i^c(t)$ can be written as,
\begin{align}
\label{proof_p2_2}
 E_i^c(t)= (\eta +1)c_i^M
 \end{align}
 
For decentralized model, according to (\ref{equation:xx5}), $\bar{E}_i^c(t)$ can be written as,
\begin{align}
\label{proof_p2_3}
\bar{E}_i^c(t)= \eta\bar{c}_i = \eta (n-1)c_i^M = M c_i^M
 \end{align}
 
In order to compare the total consumed energy per time-slot between LORD and decentralized model, let $E_i^{tc}(t)$ and $\bar{E}_i^{tc}(t)$ denote total consumed energy of nodes at time-slot $t$ for LORD and decentralized model, respectively. According to (\ref{proof_p2_1}) and (\ref{proof_p2_2}), $E_i^{tc}(t)$ is,
\begin{align}
\label{proof_p2_4}
 &E_i^{tc}(t)= \sum_{i=1}^n E_i^c (t) = \sum_{i=1}^n (\eta +1)c_i^M \nonumber
 \\
 & + 
 \min_i \Big\{   M \Big( (1 + \dfrac{m_i +1}{M})c_i ^{P} + (3+\dfrac{2}{\eta})c_i ^{M}  \Big)  - (\eta +1)c_i^M \Big\}
 \nonumber \\
 & = 
 \sum_{i=1}^n (\eta +1)c_i^M 
 \nonumber \\
& + 
 \min_i \Big\{   M \Big( (1 + \dfrac{m_i +1}{M})c_i ^{P} + (3+\dfrac{2}{\eta} - \dfrac{\eta +1}{M})c_i ^{M}  \Big) \Big\} 
 \end{align}

According to (\ref{proof_p2_3}), $\bar{E}_i^{tc}(t)$ is, 
\begin{align}
\label{proof_p2_5}
\bar{E}_i^{tc}(t)= M \sum_{i=1}^n c_i^M
 \end{align}

By considering the relationship between $c_i^M$ and $c_i^P$ as $c_i^M = \theta c_i^P$, (\ref{proof_p2_4}) can be simplified as, 
 \begin{align}
\label{proof_p2_6}
 &E_i^{tc}(t)= 
 \sum_{i=1}^n (\eta +1)c_i^M 
 \nonumber \\
& + 
 \min_i \Big\{   M \Big( \dfrac{1}{\theta}(1 + \dfrac{m_i +1}{M}) + 3+\dfrac{2}{\eta} - \dfrac{\eta +1}{M} \Big)c_i ^{M}  \Big\} 
 \end{align}
Note that by assuming $m_i \simeq \eta$ for all $i$, $\dfrac{m_i +1}{M} = \dfrac{m_i +1}{\eta(n-1)} <1$ for $n\geq10$. Therefore, we have,
 \begin{align}
\label{proof_p2_7}
E_i^{tc}(t) <
(\eta +1) \sum_{i=1}^n c_i^M  + 
 \min_i \Big\{   M \Big( \dfrac{2}{\theta} + 3+\dfrac{2}{\eta} \Big)c_i ^{M}  \Big\} 
 \end{align}
 
Now, in order to have $E_i^{tc}(t)  < \bar{E}_i^{tc}(t) $, considering (\ref{proof_p2_7}), we should have,
\begin{align}
\label{proof_p2_8}
(\eta +1) \sum_{i=1}^n c_i^M  + 
 \min_i \Big\{   M \Big( \dfrac{2}{\theta} + 3+\dfrac{2}{\eta} \Big)c_i ^{M}  \Big\} 
 < \bar{E}_i^{tc}(t) 
 \end{align}
 
Substituting $\bar{E}_i^{tc}(t) $ from (\ref{proof_p2_5}) in  (\ref{proof_p2_8}), we have,
\begin{align}
\label{proof_p2_9}
(\eta +1) \sum_{i=1}^n c_i^M  + 
 \min_i \Big\{   M \Big( \dfrac{2}{\theta} + 3+\dfrac{2}{\eta} \Big)c_i ^{M}  \Big\} 
 < M \sum_{i=1}^n c_i^M
 \end{align}
 
Simplifying (\ref{proof_p2_9}), the sufficient condition for $E_i^{tc}(t)  < \bar{E}_i^{tc}(t) $ can be written as,
\begin{align}
\label{proof_p2_10}
\min_i c_i^M < \dfrac{1- \frac{\eta +1}{\eta(n-1)}}{3+ \frac{2}{\eta} +  \dfrac{2}{\theta} } \sum_{j=1}^n c_j^M
\end{align}

If $\theta \geq1$, 
\begin{align}
\label{proof_p2_11}
\dfrac{1- \frac{\eta +1}{\eta(n-1)}}{5+ \frac{2}{\eta}} \sum_{j=1}^n c_j^M
\leq \dfrac{1- \frac{\eta +1}{\eta(n-1)}}{3+ \frac{2}{\eta} +  \dfrac{2}{\theta} } \sum_{j=1}^n c_j^M
\end{align}

Furthermore, assuming node $k$ has the lowest $c_j^M$, that is $\min_i c_i^M = c_k^M$; we have 
\begin{align}
\label{proof_p2_12}
c_k^M \leq c_j^M \hspace{2mm} \forall j \Rightarrow nc_k^M \leq \sum_{i=1}^n c_i^M 
\end{align}

Combining (\ref{proof_p2_11}) and (\ref{proof_p2_12}), we have,
\begin{align}
\label{proof_p2_13}
\dfrac{1- \frac{\eta +1}{\eta(n-1)}}{5+ \frac{2}{\eta}} nc_k^M
\leq \dfrac{1- \frac{\eta +1}{\eta(n-1)}}{5+ \frac{2}{\eta}} \sum_{j=1}^n c_j^M
\end{align}

Now, considering  (\ref{proof_p2_10}) and (\ref{proof_p2_13}), we can find a new condition for $E_i^{tc}(t)  < \bar{E}_i^{tc}(t) $ as,
\begin{align}
\label{proof_p2_14}
&c_k^M <\dfrac{1- \frac{\eta +1}{\eta(n-1)}}{5+ \frac{2}{\eta}} nc_k^M 
\Rightarrow  \nonumber \\
&n - \dfrac{n(\eta +1)}{\eta(n-1)} -5 - \frac{2}{\eta} >0
\end{align}

We define $K(n,\eta) =  n - \dfrac{n(\eta +1)}{\eta(n-1)} -5 - \frac{2}{\eta}$. Table \ref{Energy_table} shows the values of $K(n,\eta)$ for $n=$ 10 to 30 and $\eta =$ 1 to 10. As can be seen from Table \ref{Energy_table}, $K(n,\eta)>0$. Therefore, for $n \geq 10$, the condition of (\ref{proof_p2_14}) holds for any $\eta$. In other words,  for $n \geq 10$ LORD consumed less energy per time-slot compared to decentralized model, that is $E_i^{tc}(t)  < \bar{E}_i^{tc}(t)  \hspace{2mm} \forall t$.

Furthermore, note that according to aforementioned analysis, 
\begin{align}
\label{proof_p2_15}
 \bar{E}_i^{tc}(t) - E_i^{tc}(t) > K(n,\eta)
\end{align}
which means that $K(n,\eta)$ is a lower-bound on the amount of energy that LORD saves compared to decentralized model. As can be seen from Table \ref{Energy_table},
\begin{itemize}
\item For a fixed $\eta$, as the size of MDC (that is, $n$) increases, $K(n,\eta)$ increases. This means that as the size of MDC increases, LORD saves more energy in comparison with decentralized model.
\item For a fixed $n$, as $\eta$ increases, $K(n,\eta)$ increases. This means that as the demand of nodes for RD increases, LORD saves more energy in comparison with decentralized model.

\end{itemize}

\begin{table*}[b]
\begin{center}
\begin{tabular}{ cc||c|c|c|c|c|c|c|c|c|c||}
\cline{3-12} & & \multicolumn{10}{c||}{$\eta$}\\
\cline{3-12} &  & 1 & 2 & 3 & 4 & 5 & 6 & 7 & 8 & 9 &10 \\ [0.5ex] 
 \hline \hline
\multicolumn{1}{||c|}{\multirow{21}{*}{$n$} }
& 10 &     0.78 &     2.33 &     2.85 &     3.11 &     3.27 &     3.37 &     3.44 &     3.50 &     3.54 &     3.58 \\ 
\multicolumn{1}{||c|}{}
& 11 &     1.80 &     3.35 &     3.87 &     4.12 &     4.28 &     4.38 &     4.46 &     4.51 &     4.56 &     4.59 \\ 
\multicolumn{1}{||c|}{}
& 12 &     2.82 &     4.36 &     4.88 &     5.14 &     5.29 &     5.39 &     5.47 &     5.52 &     5.57 &     5.60 \\ 
\multicolumn{1}{||c|}{}
& 13 &     3.83 &     5.38 &     5.89 &     6.15 &     6.30 &     6.40 &     6.48 &     6.53 &     6.57 &     6.61 \\ 
\multicolumn{1}{||c|}{}
& 14 &     4.85 &     6.38 &     6.90 &     7.15 &     7.31 &     7.41 &     7.48 &     7.54 &     7.58 &     7.62 \\ 
\multicolumn{1}{||c|}{}
& 15 &     5.86 &     7.39 &     7.90 &     8.16 &     8.31 &     8.42 &     8.49 &     8.54 &     8.59 &     8.62 \\ 
\multicolumn{1}{||c|}{}
& 16 &     6.87 &     8.40 &     8.91 &     9.17 &     9.32 &     9.42 &     9.50 &     9.55 &     9.59 &     9.63 \\ 
\multicolumn{1}{||c|}{}
& 17 &     7.88 &     9.41 &     9.92 &    10.17 &    10.32 &    10.43 &    10.50 &    10.55 &    10.60 &    10.63 \\ 
\multicolumn{1}{||c|}{}
& 18 &     8.88 &    10.41 &    10.92 &    11.18 &    11.33 &    11.43 &    11.50 &    11.56 &    11.60 &    11.64 \\ 
\multicolumn{1}{||c|}{}
& 19 &     9.89 &    11.42 &    11.93 &    12.18 &    12.33 &    12.44 &    12.51 &    12.56 &    12.60 &    12.64 \\ 
\multicolumn{1}{||c|}{}
& 20 &    10.89 &    12.42 &    12.93 &    13.18 &    13.34 &    13.44 &    13.51 &    13.57 &    13.61 &    13.64 \\ 
\multicolumn{1}{||c|}{}
& 21 &    11.90 &    13.43 &    13.93 &    14.19 &    14.34 &    14.44 &    14.51 &    14.57 &    14.61 &    14.64 \\ 
\multicolumn{1}{||c|}{}
& 22 &    12.90 &    14.43 &    14.94 &    15.19 &    15.34 &    15.44 &    15.52 &    15.57 &    15.61 &    15.65 \\ 
\multicolumn{1}{||c|}{}
& 23 &    13.91 &    15.43 &    15.94 &    16.19 &    16.35 &    16.45 &    16.52 &    16.57 &    16.62 &    16.65 \\ 
\multicolumn{1}{||c|}{}
& 24 &    14.91 &    16.43 &    16.94 &    17.20 &    17.35 &    17.45 &    17.52 &    17.58 &    17.62 &    17.65 \\ 
\multicolumn{1}{||c|}{}
& 25 &    15.92 &    17.44 &    17.94 &    18.20 &    18.35 &    18.45 &    18.52 &    18.58 &    18.62 &    18.65 \\ 
\multicolumn{1}{||c|}{}
& 26 &    16.92 &    18.44 &    18.95 &    19.20 &    19.35 &    19.45 &    19.53 &    19.58 &    19.62 &    19.66 \\ 
\multicolumn{1}{||c|}{}
& 27 &    17.92 &    19.44 &    19.95 &    20.20 &    20.35 &    20.46 &    20.53 &    20.58 &    20.62 &    20.66 \\ 
\multicolumn{1}{||c|}{}
& 28 &    18.93 &    20.44 &    20.95 &    21.20 &    21.36 &    21.46 &    21.53 &    21.58 &    21.63 &    21.66 \\ 
\multicolumn{1}{||c|}{}
& 29 &    19.93 &    21.45 &    21.95 &    22.21 &    22.36 &    22.46 &    22.53 &    22.58 &    22.63 &    22.66 \\ 
\multicolumn{1}{||c|}{}
& 30 &    20.93 &    22.45 &    22.95 &    23.21 &    23.36 &    23.46 &    23.53 &    23.59 &    23.63 &    23.66 \\ 
\hline
\end{tabular}
\end{center}
\caption{}
\label{Energy_table}
\end{table*}

\section{Suitable Punishment Policy for Phase II}
\label{punishment}
Before the end of each time-slot, the leader should perform an auction among all participants including itself. Since the current leader of the MDC has the responsibility of performing the auction for the next time-slot, such a node can disturb our leader selection process in order to select itself as the next leader by manipulating the values offered by other. In order to avoid such misbehavior, as illustrated in Fig. \ref{figure_5}, we consider a period of time at the end of each auction process where nodes can claim that their bids are manipulated by the current leader. We also consider a consecutive period in which such claims are evaluated by other participants. 

\begin{figure}[h]
\begin{center}
\includegraphics[width=3in,keepaspectratio]{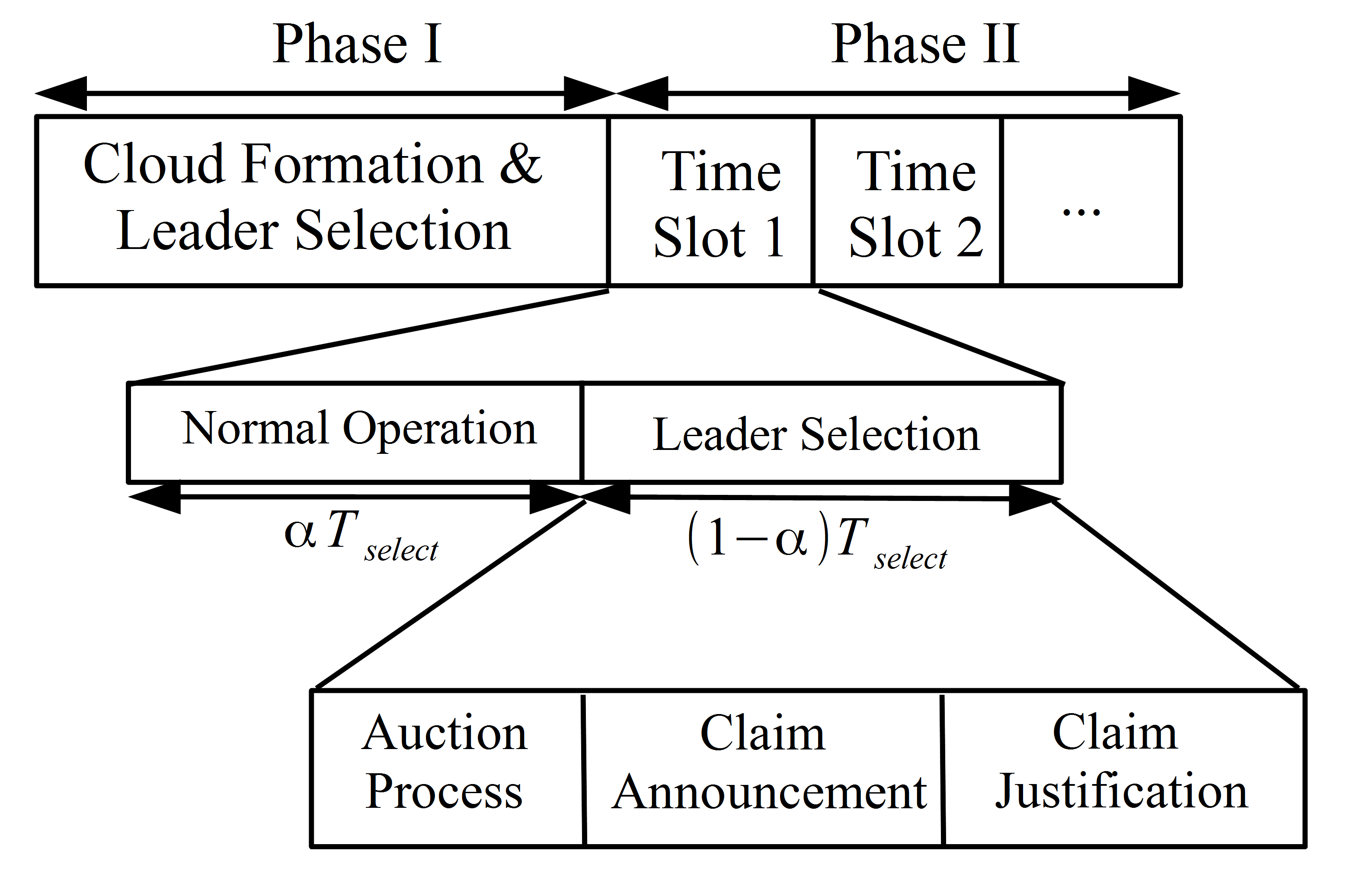}
\caption{Mobile device cloud management timing}
\label{figure_5}
\end{center}
\end{figure}

According to Fig. \ref{figure_5}, after the auction process and in claim announcement period, nodes are able to object to the result if their bids have been changed by the current leader. If there is only one node with such a claim, it will broadcast a message containing the valid value of its bid. When other nodes receive such a message, in claim justification period, they compare the bid declared by such a node with the one sent by the current leader. In the case that the claim is indeed legitimate, all MDC members except for the current leader confirm such a claim, send a message to such a node, and ask it to perform the leader selection process itself. 

If there is more than one node claiming that the current leader has changed their bids, all nodes process the received messages by comparing the values sent by the current leader with the ones claimed by these nodes. If the claims of such nodes are valid, the current leader is put into the blacklist and deprived of participating in the MDC. Then, among the claimant nodes, the node with the lowest bid is selected as the node which should perform the leader selection process. On the other hand, in the case that the claims are not valid, nodes which have received such claims send a message to the current leader, and ask it to put such nodes in the blacklist due to their fake claim. Then, as a punishment, such nodes are deprived of participating in the MDC.

Considering such punishment policy, the current leader has no incentive to deviate from the truthful strategy. Furthermore, clients have no incentive to misbehave, because they know that they will be punished if they raise an invalid objection.

Since the leader selection process imposes a cost mainly on the leader, based on (\ref{equation:xx1}), the leader receives payment in order to compensate for the cost imposed on it due to performing a leader selection process. Suitable punishment policy such as exclusion from the MDC can be used in order to make sure that the leader does not deviate from performing the auction.

\end{document}